%% file: main.tex
\documentclass[letterpaper,twocolumn,10pt]{article}
\usepackage[bottom]{footmisc}
\usepackage{usenix2019_v3}

\usepackage{tikz}
\usepackage{amsmath}
\usepackage{textcomp}

\usepackage{filecontents}

\input{headers}

\begin{document}

\date{}

\title{ Blockene: A High-throughput Blockchain Over Mobile Devices}

\author{\rm Sambhav Satija$^\star$, Apurv Mehra$^\star$, Sudheesh Singanamalla$^\dag$\thanks{Sudheesh was with Microsoft Research India while doing this work.} , Karan Grover$^\star$, \\
\rm Muthian Sivathanu$^\star$, Nishanth Chandran$^\star$, Divya Gupta$^\star$, Satya Lokam$^\star$ \\
$^\star$Microsoft Research India ~~~ $^\dag$Univ. of Washington
}

\maketitle

\input{macros}

\input{abstract}
\input{intro}

\input{background}

\input{related}
\input{architecture}

\input{threat}

\input{design}
\input{optimizations}

\input{proofs-overview}
\input{implementation}
\input{evaluation}
\input{conclusion}

\atColsBreak{\vskip10pt}
\bibliography{main}
\bibliographystyle{plain}
\section*{Appendix}
\appendix
\input{appendix}

\end{document}

%% file: headers.tex
\usepackage[english]{babel}

\usepackage{algorithm,algcompatible}

\usepackage{algpseudocode}
\usepackage{balance}
\usepackage{color}
\usepackage{xcolor}
\usepackage{nicefrac}
\usepackage{amsmath}
\usepackage{braket}

\usepackage{bm}
\usepackage{mathtools}
\usepackage{multirow}
\usepackage{bigdelim}
\usepackage{mathtools}
\usepackage{amssymb}
\usepackage{indentfirst}
\usepackage{booktabs}
\usepackage{amsthm}
\usepackage{boxedminipage}
\usepackage{float}
\usepackage{caption}
\usepackage{subcaption}
\usepackage[normalem]{ulem}
\usepackage{xpatch}
\usepackage{xspace}
\usepackage{amsmath}
\usepackage{framed,mdframed}
\usepackage{hyperref}
\usepackage{enumitem}

\definecolor{shadecolor}{rgb}{1, 0.8, 0.3}
\usepackage{graphicx}
\graphicspath{ {images/} }

\makeatletter
\xpatchcmd{\algorithmic}{\itemsep\z@}{\itemsep=0.3ex plus2pt}{}{}
\makeatother

\newlength{\saveparindent}
\newlength{\saveparskip}
\newcounter{ctr}

\newenvironment{tiret}{%
\begin{list}{\hspace{2pt}\rule[0.5ex]{6pt}{1pt}\hfill}{\labelwidth=15pt%
\labelsep=5pt \leftmargin=20pt \topsep=3pt%
\setlength{\listparindent}{\saveparindent}%
\setlength{\parsep}{\saveparskip}%
\setlength{\itemsep}{0pt} }}{\end{list}}

\newenvironment{newenum}{%
\begin{list}{{\rm \arabic{ctr}.}\hfill}{\usecounter{ctr}\labelwidth=15pt%
\labelsep=6pt \leftmargin=20pt \topsep=.5pt%
\setlength{\listparindent}{\saveparindent}%
\setlength{\parsep}{\saveparskip}%
\setlength{\itemsep}{5pt} }}{\end{list}}

\definecolor{mypink}{rgb}{1,0.2,0.4}

\newcommand{\commround}{N}
\newcommand{\sortition}{\mathsf{Sortition}}

\newtheorem{theorem}{Theorem}

\newtheorem{lemma}{Lemma}

\newtheorem{definition}{Definition}
\newtheorem{property}{Property}
\newtheorem{corollary}{Corollary}

\newcommand{\namedref}[2]{\hyperref[#2]{#1~\ref*{#2}}\xspace}
\newcommand{\lemmaref}[1]{\namedref{Lemma}{lemma:#1}}

\newcommand{\algoref}[1]{\namedref{Algorithm}{algo:#1}}
\newcommand{\sk}{\mathsf{sk}}
\newcommand{\vk}{\mathsf{vk}}
\DeclareMathOperator{\localstate}{\mathsf{LS}}

  \newcommand{\eps}[0]{\ensuremath{\varepsilon}}
  \let\epsilon\eps

  \newcommand{\cF}{\ensuremath{{\mathcal F}}\xspace}

 \newcommand{\cK}{\ensuremath{{\mathcal K}}\xspace}

\newcommand{\sn}{\mathcal{S}}

\newcommand{\hmaj}{\alpha}

\newcommand{\corrsn}{\gamma}
\newcommand{\fanout}{m}
\newcommand{\fanoutnum}{25}
\newcommand{\allkeys}{\cK}
\newcommand{\deep}{d}
\newcommand{\hashsz}{h}
\newcommand{\keys}{\mathsf{keys}}
\newcommand{\collide}{\theta}
\newcommand{\bnum}{B}
\newcommand{\hv}{\mathsf{hv}}

\newcommand{\tree}{T}
\newcommand{\cutpt}{a}
\newcommand{\fnodes}{\cF}
\newcommand{\mem}{\mathsf{mem}}
\newcommand{\idlist}{\mathsf{IDList}}

\newcommand{\key}{\mathsf{key}}
\newcommand{\val}{\mathsf{val}}

\newcommand{\snnum}{S}
\newcommand{\commnum}{n}
\newcommand{\mobilenum}{M}
\newcommand{\commnumlow}{\commnum^*}
\newcommand{\commnummean}{\commnum'}
\newcommand{\commnumhigh}{\tilde{\commnum}}
\newcommand{\heavynum}{\commnum_{H}}

\newcommand{\goodnum}{\commnum_{g}}
\newcommand{\badnum}{\commnum_{b}}
\newcommand{\gap}{\mathsf{gap}}
\newcommand{\kl}[2]{\mathsf{KL}(#1, #2)}
\newcommand{\secparam}{\kappa}
\newcommand{\badnummax}{\tilde{n_b}}

\newcommand{\hash}{\mathsf{Hash}}
\newcommand{\mysign}{\mathsf{Sign}}
\newcommand{\myBlock}{\mathsf{Block}}
\newcommand{\wl}{\mathsf{WL}}
\newcommand{\p}[1]{^{(#1)}}
\newcommand{\vrf}{\mathsf{VRF}}

\newcommand{\block}{\mathcal{B}}

\newcommand{\proposal}{\mathsf{Proposal}}
\newcommand{\win}{\mathsf{win}}
\newcommand{\initial}{\mathsf{initial}}
\newcommand{\nullv}{\mathsf{NULL}}
\newcommand{\out}{\mathsf{out}}
\newcommand{\goodthresh}{\mathsf{T^*}}

\newcommand{\subblock}{\mathsf{SB}}
\newcommand{\hc}{\mathsf{HC}}

\newcommand{\idpk}{\mathsf{Id\text{-}PK}}
\newcommand{\gspk}{\mathsf{GS\text{-}PK}}
\newcommand{\tk}{\mathsf{tk}} %
\newcommand{\cert}{\mathtt{cert}}
\newcommand{\GSRoot}{\mathsf{GSRoot}}
\newcommand{\cooloff}{\ensuremath{40}}

%% file: macros.tex
\newcommand{\eg}{\textit{e.g.}}
\newcommand{\ie}{\textit{i.e.}}
\newcommand{\etal}{\textit{et al.}}
\newcommand{\etc}{\textit{etc.}}
\newcommand{\adhoc}{\textit{ad hoc}}

\newcommand{\blockene}{\textit{Blockene}}
\newcommand{\citizens}{\textit{Citizens}}
\newcommand{\politicians}{\textit{Politicians}}
\newcommand{\citizen}{\textit{Citizen}}
\newcommand{\politician}{\textit{Politician}}
\newcommand{\tpool}{\textit{tx\_pool}}

\newcommand{\tpools}{\textit{tx\_pools}}
\newcommand{\pc}{\textit{PCR}}
\newcommand{\pcs}{\textit{PCRs}}
\newcommand{\pclong}{\textit{precursor}}

\newcommand{\pcpool}{\textit{pcr\_pool}}

\newcommand{\pcpools}{\textit{pcr\_pools}}

\newcommand\dg[1]{{\color{blue} DG: {#1}}}

\newcommand{\revision}[1]{\textcolor{black}{#1}}
\newcommand{\secondrevise}[1]{\textcolor{black}{#1}}

%% file: abstract.tex
\begin{abstract}
{\em
We introduce \blockene, a blockchain that reduces resource usage at member nodes by orders of magnitude, requiring only a smartphone to participate in block validation and consensus.  Despite being lightweight, \blockene\ provides a high throughput of transactions and scales to a large number of participants.  \blockene\ consumes negligible battery and data in smartphones, enabling millions of users to participate in the blockchain without incentives, to secure transactions with their collective honesty.  \blockene\ achieves these properties with a novel split-trust design based on delegating storage and gossip to untrusted nodes.

We show, with a prototype implementation, that \blockene\ provides throughput of 1045 transactions/sec, and runs with very low resource usage on smartphones, pointing to a new paradigm for building secure, decentralized applications.
}
\end{abstract}

%% file: intro.tex
\section{Introduction}
\label{sec-intro}

Blockchains provide a powerful systems abstraction: they allow mutually untrusted entities ({\em members}) to collectively manage a {\em ledger} of transactions in a decentralized manner.

\revision{All blockchains today require member nodes to run powerful servers with significant network, storage, and compute resources.  Blockchains based on {\em proof-of-work}~\cite{Bitcoin,Ethereum} push resource usage to an extreme, requiring significant compute for puzzle-solving, but even consortium blockchains~\cite{androulaki2018hyperledger} and blockchains based on {\em proof-of-stake}~\cite{AlgoRand} incur significant network and storage costs to keep the blockchain up to date at a high transaction throughput.    Blockchains today} are therefore limited to use-cases where members have a strong incentive to participate, and can hence afford the high resource cost.  For example, in consortium blockchains~\cite{androulaki2018hyperledger}, business efficiency improves, while in cryptocurrencies~\cite{Bitcoin, AlgoRand}, members earn currency.

Interestingly, the high resource requirement of blockchains also weakens {\em reliability} for several real-world applications.  Blockchains require that majority (typically two-thirds) of members are honest, a property that is easier to guarantee when a large number of members participate.  However, wide-scale adoption of a blockchain is hard given the high resource requirement, especially in scenarios where members do not have a direct incentive to participate.  Not surprisingly, public blockchains with high membership today target cryptocurrencies~\cite{Bitcoin,Ethereum}.

In this paper, we present \blockene \footnote{Named after {\em Graphene}, one of the lightest and strongest materials.}, an ultra-lightweight, large scale blockchain that provides high throughput for real-world transactions.  By being lightweight and scalable, it enables wide-scale adoption by millions of users.
By enabling large scale of participation, \blockene\ makes it plausible to assume honest-majority.   By being high-throughput, \blockene\ supports real-world transaction rates.

The key breakthrough in \blockene\ is that instead of requiring members to run powerful servers, \blockene\ is \revision{the first blockchain that enables members to participate as first-class citizens in consensus even while running} on devices as lightweight as smartphones, lowering cost by orders of magnitude.   

\noindent{\bf Network:} Blockchains rely on peer-to-peer gossip between members; at a high transaction rate, gossip would require tens of GBs of data transfer per day; \blockene\ requires only about 60MB of data transfer per day on a smartphone.   

\noindent{\bf Storage:} Member nodes in blockchains keep a copy of the entire blockchain (terabytes at high-throughput); in \blockene, members incur only a few hundred MBs of storage.  

\noindent{\bf Compute:} \revision{Even the gossip cost of typical blockchains} would drain battery on mobile nodes; \blockene\ ensures that battery drain is less than 3\% per day.  

Thus, a user incurs no perceptible cost while running \blockene.  As the low resource usage in \blockene\ makes it feasible even in a smartphone, \blockene\ can also run on desktops, with much lighter resource usage than state-of-the-art.

\blockene\ achieves three conflicting properties: large scale of participation, high throughput, and lightweight resource usage, catering to even scenarios where there is no direct incentive (\eg, altruistic participation), and  handling transactions across variety of use-cases including those on public funds.  A comparison of \blockene\ with other blockchain architectures is depicted in Table~\ref{tab-taxonomy}.

\noindent{\bf Example application: Audited Philanthropy.}  Charitable donations to non-profits are in excess of USD 500 billion {\em annually} worldwide~\cite{ngo-scale,ngo-scale-europe,ngo-scale-india}.  However, from a donor's perspective, the lack of transparency on the end-use of funds makes donations vulnerable to sub-optimal use or mismanagement by non-profits, especially in regions where regulatory enforcement is ineffective or crippled by corruption.  A system that provides a public, end-to-end trail of funds from the donor to the end beneficiary, will exert market pressure on non-profits, besides motivating donors.  A blockchain can provide such tracking, but given the scale of funds involved, a small consortium of members cannot be trusted with operation of the blockchain.  Ideally, such a blockchain should be jointly controlled by millions of citizens altruistically.  Similar requirements arise in government/public spending.

\noindent{\bf Key techniques in Blockene: }
\blockene\ adopts a novel system design based on a split-trust architecture \revision{with a new security model}.
There are two types of nodes in \blockene: \citizens\ and \politicians.
\citizens\ run on smartphones and are the real members of the blockchain, i.e., they have voting power in consensus protocol; hence we assume that two-thirds of the \citizens\ are honest (a reasonable assumption with millions of \citizens).
On the other hand, \politicians\ run on servers and are untrusted, i.e., do not participate in consensus.
\politicians\ are fewer in number (few hundreds), and we require {\em only 20\%} of them to be honest.
Although \politicians\ do the heavy work such as storing the blockchain, our protocols ensure that \citizens\ can detect and handle malicious behavior even if 80\% of the \politicians\ collude with the one-third of malicious \citizens.   
\citizens\ deal with high dishonesty of \politicians\ by using a new primitive called {\em replicated verifiable reads}: the \citizen\ reads the same data from multiple \politicians\ and can get the correct value even if one (out of, say, $\fanoutnum$) is honest.

\citizens\ perform transaction validation, and decide on the block and resulting global state to commit,  by running Byzantine consensus.  To make consensus feasible with millions of \citizens, \blockene\ borrows an idea from Algorand~\cite{AlgoRand} (modified to make it battery-friendly), where a different random committee of (\textasciitilde2000) \citizens\ is cryptographically chosen to run consensus for each block.  \secondrevise{Unlike Algorand, \blockene\ exposes the set of committee members a few minutes before their participation: this enables \blockene\ to reduce data and battery cost at \citizens.  While this may appear to increase the window for a targeted attack on the committee, we discuss in \S~\ref{sec:threat} why this is not a serious concern.}

To keep storage/communication costs at the \citizens\ low, only \politicians\ store the blockchain and the global state (\ie, key-value pairs), 
freeing \citizens\ from gossiping all blocks (\textasciitilde 50GB/day).  \citizens\ only read a small subset of data from \politicians\ (\eg, key-values for transactions for the current block), and write out the new block.
Further, because \politicians\ are untrusted, \citizens\ cannot rely on the correct latest values returned by them for, say, a given key.   \blockene\ uses a novel technique of {\em sampling-based Merkle tree read/write} that reduces communication cost while ensuring tolerance to 80\% malicious \politicians.

When in the committee, \citizens\ reduce their communication cost by not gossiping directly, but through \politicians; data written by a \citizen\ gets gossiped among \politicians, and interested \citizens\ read from \politicians.

As participation in \blockene\ is lightweight, the system needs to protect against Sybil attacks~\cite{Sybil}; preventing an adversary from spinning up lots of virtual nodes to get disproportionate voting share.  To thwart such attacks, \blockene\ requires the participant identity to be certified by the trusted hardware (TEE) available in most smartphones~\cite{android-tee, apple-tee}, and enforces that each TEE can have at most one active identity on the blockchain, thus raising the economic cost of participation to the cost of a unique smartphone.

To limit damage that 80\% malicious \politicians\ can cause to performance, \blockene\ employs several techniques to restrict their ability to lie.   First, we use a technique called {\em pre-declared commitments} to  make some malicious behaviors detectable.  Second, to perform gossip among \politicians\ reliably and efficiently despite 80\% dishonesty, we introduce a novel technique called {\em prioritized gossip}.   
These techniques reduce cost at \citizens, enabling \blockene\ to achieve high throughput despite running on smartphones.

We have built a prototype of \blockene; the \citizen\ node is implemented as an Android application, and \politician\ node is implemented as a cloud server.  We evaluate \blockene\ along various dimensions, and show that it achieves good transaction throughput of 1045 transactions/sec (6.8 MB/min) while ensuring a commit latency of 270s in the 99th percentile.  We also demonstrate very little data use (61 MB/day) and battery use (3\%/day) at \citizens.

The key contributions of this paper are as follows:
\begin{tiret}
\setlength\itemsep{0.2em}
\item We present the first blockchain system \revision{where nodes can participate as first-class members in consensus while running} on devices as lightweight as smartphones, supporting high scale of members and high throughput.

\item We present a novel split-trust design \revision{with a new security model comprised of} resource constrained \citizens\ (honest majority) and resource heavy \politicians\ (dishonest majority), and \citizens\ performing validation and consensus by offloading heavy work to untrusted \politicians\ in a verifiable way.

\item  We make several novel optimizations (e.g., pre-declared commitments, sampling-based Merkle tree read/write, prioritized gossip) that achieve good performance despite 80\% malicious \politicians.

\item With a thorough theoretical analysis, we prove that \blockene\ satisfies \emph{safety, liveness,} and \emph{fairness}.

\item We perform a thorough empirical evaluation of this architecture, demonstrating its feasibility as a shared scalable blockchain service.

\end{tiret}

The rest of the paper is structured as follows: In \S~\ref{sec-background}, we provide a background on blockchains, \revision{and discuss existing blockchain architectures in \S~\ref{sec-related}}.  \S~\ref{sec-overview} provides an overview of \blockene, and its threat model, and \S~\ref{sec-design} presents its design.  We discuss optimizations for resource-heavy steps in \S~\ref{sec-optimizations}, \revision{present an overview of safety and liveness proofs in \S~\ref{sec:proofs},} and describe the implementation in \S~\ref{sec-implementation}.  We evaluate \blockene\ in \S~\ref{sec-evaluation}, and conclude (\S~\ref{sec-conclusion}). %

%% file: background.tex
\vspace{-0.1in}
\section{Background}
\label{sec-background}

In this section, we discuss the key principles and abstractions in a blockchain, and its applications. 
\vspace{-0.1in}

\subsection{Basic properties}

A blockchain is a {\em distributed ledger} of transactions.   Without a trusted authority (\eg, a bank) managing the ledger, a group of mutually untrusted parties collectively validate transactions, and maintain a consistent ledger, provided at least a threshold of participants (\eg, two-thirds) are honest. A blockchain must provide {\em safety}, {\em liveness}, and {\em fairness}.  Safety ensures that honest participants have a consistent view of the ledger.  Liveness ensures that malicious participants cannot indefinitely stall the blockchain by preventing new block additions. %
Fairness ensures that all valid transactions submitted to the blockchain get eventually committed.

\vspace{-0.1in}

\subsection{Building blocks}
\label{sec:bb}

A blockchain is a replicated, peer-to-peer distributed system built on the following basic primitives:

\noindent {\bf Merkle tree for Global State: } A key part of a blockchain is the {\em global state} database that tracks keys and their current values.  This global state is managed in a tamper-proof manner, typically using a Merkle tree where the leaf nodes contain the key-value pairs, while each intermediate node contains a hash of the concatenated contents of child nodes.    
The root is a single hash value that represents the entire state.  An update of a key requires recomputation of hashes only along the {\em path} from that leaf to the root.  Given the root, the value of any key can be proved by a path of valid hashes to the root. 

\noindent {\bf Signed transactions:}  The basic unit of work in a blockchain is a {\em transaction}.  A transaction reads and updates a few keys in the global state (\eg, transfer \$1000 from Alice to Bob).  To be valid, (a) the transaction must be signed (b) the user signing the transaction must have access to the keys (c) ``semantic'' integrity must pass (\eg, cannot overspend).

\noindent {\bf Cryptograpic linkage:}  A blockchain is a list of {\em blocks}.   A block is a list of transactions.  The ordering of blocks is ensured by a {\em cryptographic linkage}; every block embeds the cryptographic hash of the previous block's contents.

\noindent {\bf Gossip:}   Participants in a blockchain exchange state with each other in a peer-to-peer fashion.  For example, when a new block gets committed to the ledger, it must be sent to other members.  This communication happens through multi-hop {\em gossip}, with eventual consistency.

\noindent {\bf Consensus Protocol:} The key primitive in blockchains is a {\em distributed consensus}  protocol that handles {\em Byzantine failures} (\eg, PBFT~\cite{PBFT}, Nakamoto~\cite{Bitcoin}, or BBA~\cite{AlgoRand}), as minority of participants could be malicious.
Byzantine consensus requires at least 2/3rd participants to be honest, and requires several {\em rounds} of communication.

%% file: related.tex
\section{Comparison with Existing Blockchains}
\label{sec-related}

\revision{
In this section, we present a brief survey of related work on existing blockchain architectures.   \blockene\ provides three properties: lightweight resource usage, large scale of participation, and high transaction throughput.  We use the same three dimensions to compare \blockene\ with related work.
}

\subsection{Resource usage by member nodes}

\revision{
Existing blockchains span a wide spectrum in resource usage by participating member nodes, depending on the mechanism used for consensus.   We first discuss compute cost incurred by members, and then the network and storage cost.
}

\noindent{\bf Compute Cost. }
\revision{
In terms of compute cost, the most expensive are blockchains based on Nakamoto consensus~\cite{Bitcoin}, also referred to as {\em proof-of-work}; examples are Bitcoin~\cite{Bitcoin} and Ethereum~\cite{Ethereum}.  In Nakamoto consensus, the first member node to solve a compute-intensive cryptographic puzzle is chosen as the winner in committing a new block.   Such blockchains therefore require heavy compute resources at member nodes.
}

\revision{
In order to address the high compute (and energy) costs of proof-of-work blockchains, two popular alternative architectures have emerged.  The first is {\em consortium blockchains} (\eg, HyperLedger~\cite{androulaki2018hyperledger}), which, by limiting the blockchain membership to a small number of nodes, can run traditional Byzantine consensus algorithms, instead of the compute-intensive proof-of-work based consensus.   The second architecture is {\em proof-of-stake} blockchains, which tie the voting power of a member node with the amount of money the member node has on the blockchain.  Examples of these blockchains are Algorand~\cite{AlgoRand}, Ouroboros~\cite{ouroboros1, ouroboros2}, PeerCoin~\cite{Peercoin}, \etc.  Inherently, proof-of-stake blockchains target cryptocurrency applications where such a ``stake'' is meaningful.
}

\noindent{\bf Network and Storage cost. }
\revision{
While the above two architectures, \ie, consortium blockchains and proof-of-stake blockchains, address the raw compute cost of member nodes, they are still too expensive for smartphones.  In particular, they are heavy on network and storage resources, as they require the member nodes {\em to  be always up-to-date} with the ``current'' state of the blockchain.   Given the high transaction rate (1000s of transactions per second) that such blockchains enable,  replication of the entire state across member nodes is expensive: at 1000 transactions/sec, the blockchain would commit roughly 9GB per day, which needs to be gossiped across member nodes, resulting in a network cost of roughly 45 GB/day (assuming a gossip fanout of 5 neighbors) that every member node has to incur.  Further, such a blockchain would consume terabytes of storage on member nodes, as every member node stores a local copy of the blockchain.  
}

\revision{
Even blockchains that target smartphones~\cite{suankaewmanee2018performance} adopt the same philosophy of member nodes staying up to date, and thus incur the network and storage overheads.
}

\revision{
Some blockchains address storage cost by {\em sharding}.  OmniLedger~\cite{omniledger} is a recent blockchain that allows participants to only store a {\em shard} of the blockchain. It uses a variant of Byzcoin~\cite{byzcoin} for fast consensus. RapidChain~\cite{rapidchain} also uses sharding to reduce storage cost. Both these works scale only to a few thousand participants and also require participants to store a large fraction ($\frac{1}{3}$ or $\frac{1}{16}$) of the entire blockchain.
}

\noindent{\bf Lightweight but Incapable Nodes. }
\revision{
A class of ``lightweight'' blockchains adopt an approach of ``unequal members'': only the first-tier, resource-heavy members participate in consensus and have voting power, while the second-tier members simply serve as read-only query frontends, and do not participate in consensus.    In such a model, the ``majority-honest'' property must be met purely by the heavy nodes, as light nodes do not contribute to security.   Not surprisingly, given the limited responsibility, the ``light'' nodes don't consume much resources.  An example of this architecture is the separation between light and heavy nodes  in Ethereum~\cite{Ethereum-Light}. 
}

\noindent{\bf Blockene. }
\revision{
In contrast, \blockene,  achieves lightweight resource usage for {\em first-class members} that participate in consensus and block validation.   Further, unlike Ethereum which depends on honest majority among heavy nodes (only heavy nodes can vote), \blockene\ tolerates up to 80\% of the ``heavy'' nodes (\ie, \politicians) being corrupt.   Members in \blockene\ require only a smartphone and negligible\footnote{Cellular data costs in several countries are much cheaper than in the US~\cite{data-costs}; in US/Europe, users are on WiFi/broadband most of the day.}  data transfer (< 60 MB/day, \ie, three orders of magnitude lower) and negligible compute (battery use of <3\% per day).  It achieves this by enabling member nodes to operate with minimal state needed for committing a particular block, and performing work only a few times a day, \ie, not striving to stay up-to-date always.
}

\subsection{Scale of participation}

\revision{
As the security of a blockchain fundamentally relies on a majority of the participating members being honest, blockchains need to protect against collusion of a large number of participants.  Consortium blockchains~\cite{androulaki2018hyperledger} carefully structure the blockchain for a particular business process, such that members have a shared incentive in the success of the blockchain.  It is sometimes infeasible/hard to structure a consortium with the above guarantee; in the philanthropy example, if a small number of members are in control of the blockchain, they may collude to, say, facilitate siphoning of donations meant for the poor.   Moreover, a consortium blockchain is intricately tied to a specific business process among a set of entities, resulting in high setup and operational overhead, besides limiting inter-operability.   
}

\revision{
Another approach to guard against collusion among majority, is to enable large scale participation; by onboarding a large number of participants (say millions), majority-collusion can be made hard and unlikely.  Most ``public'' blockchains such as Bitcoin~\cite{Bitcoin}, Ethereum~\cite{Ethereum}, and Algorand~\cite{AlgoRand} enable large-scale participation.   Blockene also supports a large number of participants, but unlike most public blockchains today that target cryptocurrencies, \blockene\ is not tied to cryptocurrency (\eg, no proof-of-stake), but enables generic business transactions.     Unlike consortium blockchains, \blockene\  can additionally enable real-world scenarios where there is potential for collusion among a small number of members.
}

\subsection{Transaction throughput}

\revision{
Public blockchains based on proof-of-work are low in throughput (\textasciitilde 4-10 transactions/sec).
Proof-of-stake based Algorand~\cite{AlgoRand} is the first public blockchain with \textasciitilde 1000 transactions/sec\footnote{Assuming 100-byte transactions and 2.2 MB in 20s, 10MB blocks @750MB/hr.} 
Consortium blockchains, due to low scale of participants and traditional consensus (e.g., PBFT),  provide 1000s of transactions/sec.   Similar to Algorand, \blockene\ also provides a high transaction throughput.  By not being tied to cryptocurrency applications, \blockene\ can serve traditional business applications similar to consortium blockchains.
}

\begin{table}[htbp]
\begin{small}
\centering

\begin{tabular}{| c | c | c | c | c |}
\hline
{\em Blockchain}                       & {\em Scale of} & {\em Trans.} & {\em Cost} & {\em Incentive} \\
& {\em members} & rate &  & {\em needed?} \\
 & & & & \\
\hline
Public           &  {\bf Millions}  &  4-10 /sec. &  Huge &  Yes                         \\
(\eg, Bitcoin) &               &                    &  (PoW) &              \\

\hline

Consortium &  Tens      &  {\bf 1000s /sec.} & High           & Yes                         \\
(\eg, ~\cite{androulaki2018hyperledger})        &               &                      &                              &             \\

\hline

Algorand~\cite{AlgoRand} & {\bf Millions}  & {\bf 1000-2000/sec.} & High           &  Yes                      \\

\hline

{\bf Blockene}                       &  {\bf Millions}     &  {\bf 1045 /sec.}       &  {\bf Tiny}    &  {\bf No}                          \\

\hline

\end{tabular}
\caption{Comparison of blockchain architectures. }
\label{tab-taxonomy}

\end{small}
\end{table}

\vspace{-0.1in}
\subsection{Incentives to Participants}

Because of high resource cost (compute, network, or storage), existing blockchains need an incentive for participants (\eg, mining coins in cryptocurrencies, or business efficiency in consortiums).  Blockchains that depend on such incentives cannot work for applications such as philanthropy (\S~\ref{sec-intro}).  To scale without incentives and to enable {\em altruistic} participation, the cost of participation has to be negligible.

Table~\ref{tab-taxonomy} compares  blockchain architectures along these dimensions.   \blockene\ is the first blockchain to achieve all of the above:  scale, throughput, and low cost.   With low cost,  \blockene\ supports real-world use-cases even where participants do not have a direct incentive, but are altruistic to run a background app with negligible battery and data usage.

\subsection{Other related work}
\revision{The committee-based consensus in \blockene\ is heavily inspired by Algorand~\cite{AlgoRand}; }
Like \blockene, Algorand also does not allow forks to occur, and one consistent view of the blockchain is always maintained.   \revision{There is a tradeoff between Algorand and \blockene\ on the resilience to two kinds of targeted attacks (described in \S~\ref{sec:threat} para 1).  HoneyBadger~\cite{honeybadger} is a recent system designed for consortium blockchains with O(100) participants.  IOTA~\cite{iota1, iota2} is another distributed ledger system, but currently relies on a centralized co-ordinator for consensus.
}

Among proof-of-work-based blockchains, the  most closely related work to \blockene\ is Hybrid consensus~\cite{hybridconsensus}.  Similar to Algorand (and \blockene), Hybrid consensus periodically selects a group of participants and does not allow the adversary to corrupt nodes during the ``participant selection interval''. However it has a long selection interval (of about 1 day) and is also open to the possibility of forks.

\if 0

\noindent\textbf{Proof of Work.} 
These require participants to solve ``crypto-puzzles'' in order to commit transactions and hence are very compute intensive. They are also susceptible to forks and have a poor transaction throughput.
While the most popular proof of work based protocol is Bitcoin~\cite{Bitcoin}, the  most closely related work to \blockene\ is Hybrid consensus~\cite{hybridconsensus}.

\noindent\textbf{Other recent works.} OmniLedger~\cite{omniledger} is a recent sharding-based blockchain protocol that aims to allow participants to only store part of the blockchain. It uses a variant of Byzcoin~\cite{byzcoin} for fast consensus. RapidChain~\cite{rapidchain} also uses sharding to reduce storage cost for every participant. Both these works scale only to a few thousand participants and also require participants to store a large fraction ($\frac{1}{3}$ or $\frac{1}{16}$ of the entire blockchain).
\fi

%% file: architecture.tex
\section{Architecture Overview}
\label{sec-overview}

In this section, we first introduce our two-tier architecture that achieves \revision{the three conflicting properties of lightweight resource usage, large scale of participation, and high transaction throughput}.  We then discuss the threat model of \blockene.

\subsection{Two-tier Architecture}
\label{sec:architecture}

\blockene\ employs a novel two-tier architecture with asymmetric trust.  This architecture is depicted in Figure \ref{fig:blockenearch}. 

 There are two kinds of nodes in \blockene: \citizens\ and \politicians.  \citizens\ are resource-constrained (\ie, run on smartphones), are large in number (millions), and are the only entities having voting power in the system (\ie, participate in consensus). \politicians\ are powerful and run servers (similar to existing blockchains like Algorand), and are lot fewer in number (low hundreds), but they do not have voting power.   \politicians\ only execute decisions taken by \citizens, and cannot take any decisions on their own.  

The low resource usage enables a large number of \citizens\ to participate without incentives, while \politicians\ being few in number, will be run by large entities that have interest in the particular use-case (\eg, in the audited philanthropy case, large donors and foundations).

As \citizens\ participate in consensus, at least two-thirds of \citizens\ are required to be honest, while others can be malicious and collude.  This is reasonable as \blockene\ allows millions of \citizens, making large-scale corruption hard.     However, \politicians\ enjoy much lower trust.  \blockene\ only requires 20\% of politicians to be honest; the remaining 80\% of the politicians can be malicious and collude among themselves, and with one-third malicious \citizens.

\begin{figure}
\centering{
\includegraphics[width=0.8\linewidth]{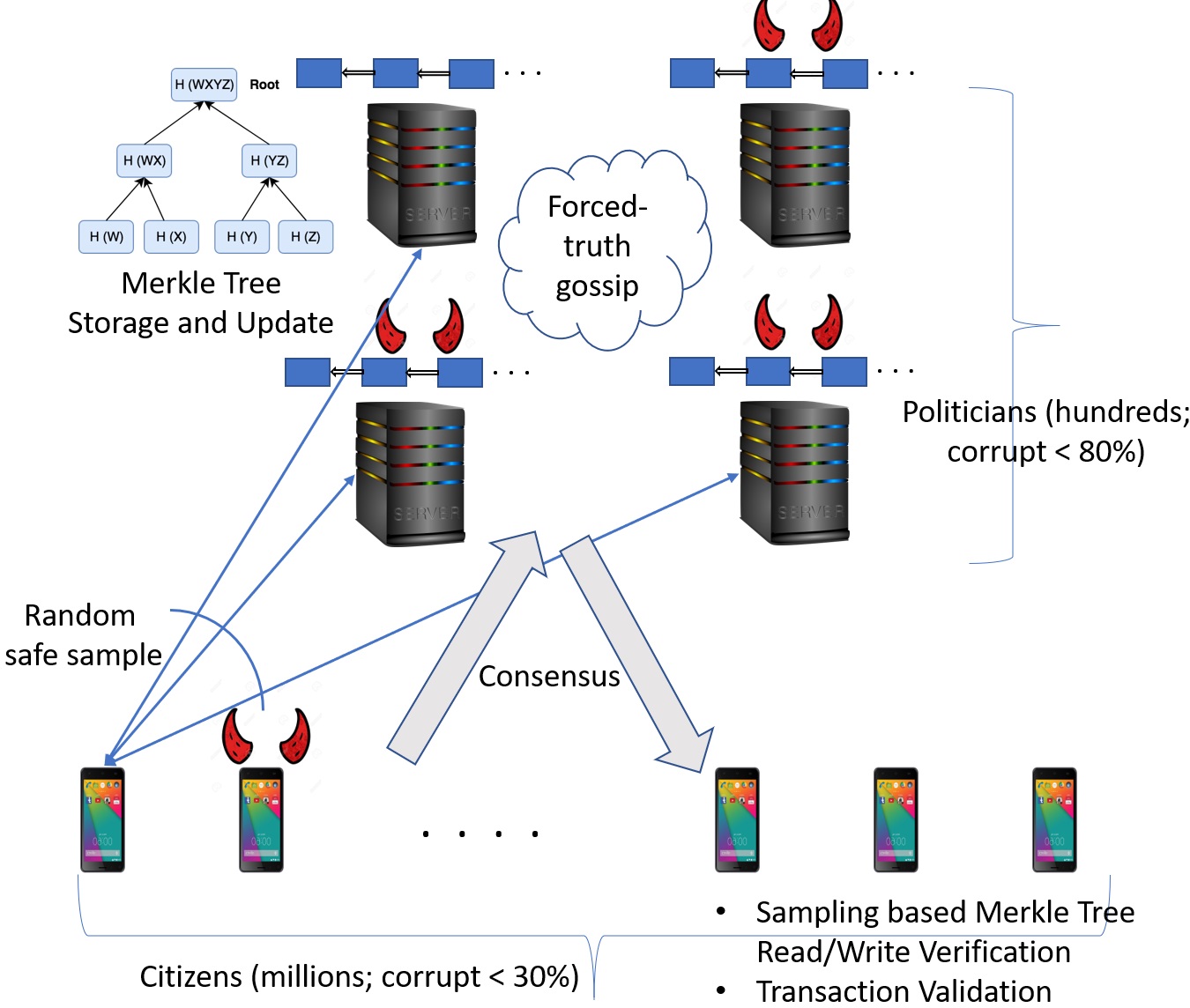}
\caption{\blockene's architecture}
\label{fig:blockenearch} 
}
\end{figure}

\subsubsection{Offloading work to \politicians}
\label{subsec-safe-sample}

Intuitively, given the two-tier architecture, \citizens\ can offload expensive responsibilities such as storage and communication to \politicians.  However, as 80\% of \politicians\ are corrupt, a write made by a \citizen\ could just be dropped by a \politician\ or, a read could return incorrect value.    
To get useful work done out of \politicians\, \blockene\ uses a novel mechanism of replicated reads and writes.   
Reads and writes by \citizens\ to \politicians\ happen with a random {\em safe sample} of \politicians.  
The size of this sample is fixed such that with high probability, at least one \politician\ in the sample is honest (\eg, for a sample size of 25, this probability is $1 - (0.8)^{25} = 99.6\%$).  \blockene\ is resilient to a small number of \citizens\ (0.4\%) picking all dishonest \politicians.

\vspace{-0.1in}
\subsubsection {Division of responsibilities}

We now describe how the \citizens\ and \politicians\ collaborate to perform the various standard blockchain tasks:

\noindent{\bf Storage:}  In a traditional blockchain, every participant keeps  a replica of the entire blockchain, but \citizens\ in \blockene\ cannot afford to store TBs of data.  In \blockene, only \politicians\ store the ledger and the global state (\ie, database of key-values \S~\ref{sec-background}).  \citizens\ read subsets of this data from \politicians\ as needed.  The only state \citizens\ store (and periodically update) is a list of valid \citizen\ identities (\S~\ref{subsec-latest-block}).

\noindent{\bf Transaction Validation:}  As \citizens\ are the actual participants in consensus, they validate transactions, ensuring that transactions are signed, and have semantic integrity (\eg, no double-spending).  To perform validation, \citizens\ read transactions from \politicians, and lookup latest values of the keys referenced in them, from the global state with \politicians.  \citizens\ then propose a block with valid transactions. 

\noindent{\bf Gossip:}  
To ensure that all honest participants agree on the state of the blockchain, participants need to gossip among themselves.  However, as discussed in \S~\ref{sec-related}, direct gossip among \citizens\ is expensive.  \blockene\ solves this problem by having \citizens\ gossip through \politicians.   When a \citizen\ needs to broadcast information to other \citizens, it sends a message to a {\em safe sample} of \politicians.  \politicians\ then gossip data among themselves; they can afford to do so because they have good network connectivity. Other \citizens\ then perform a replicated read from the \politicians\ {\em when they need to, \eg, when they are in the committee}\footnote{Direct gossip among \citizens\ would require {\em all} \citizens\ (\ie, including those outside the committee) to participate in gossip of all data.}.  For gossip through \politicians, we need the guarantee that a message that reaches one honest \politician\ always reaches all other honest \politicians\ via gossip,  a challenging property when 80\% of the \politicians\ are malicious; our custom gossip protocol is described in \S~\ref{sec:forced-truth-gossip-main}).  Thus, we achieve the same semantics as direct gossip among \citizens, but with minimal network load on \citizens.

\noindent{\bf Consensus:}  \citizens\ participate in consensus by performing gossip through \politicians.  
Given the large scale of \citizens, all \citizens\ cannot participate in consensus.  
Instead, we cryptographically select a random {\em committee of citizens} (roughly 2000 members) for each block (\S~\ref{sec:comm-sel-main}).

%% file: threat.tex
\vspace{-0.1in}
\subsection{Threat Model}
\label{sec:threat}

\secondrevise{While our threat model is similar to Algorand~\cite{AlgoRand},  there is a tradeoff between Algorand and \blockene\ on the resilience to targeted attacks.  On one hand, Algorand is based on proof-of-stake, which allows an adversary infinite time to target nodes with higher stake (who will appear in the committee more frequently); \blockene\ avoids this attack, as all \citizens\ have equal votes.   On the other hand, Algorand protects the secrecy of the committee members until they perform their role, but \blockene\ exposes their identities a few minutes (1-2 blocks) {\em before} they participate.  To conserve battery, \citizens\ normally poll  \politicians\ for current state of the blockchain roughly every 10 blocks (~\ref{sec:comm-sel-main}), but when they are going to be in the committee, will poll again shortly (\eg, 1 block) before their expected turn, thus exposing their identity to malicious \politicians.  This potentially provides a window for a targeted attack (\secondrevise{\eg, by bribing the committee: \S~\ref{subsec-citizen-attack}}).
}

\vspace{-0.1in}
\subsubsection{Attack vector of \citizens}
\label{subsec-citizen-attack}

\noindent{\bf \secondrevise{Bribing attack on \citizens: }}
\secondrevise{
As \blockene\ implicitly exposes the public keys of the committee a few (\eg, 2) minutes  in advance, an adversary could in theory perform a targeted attack by bribing a sufficient number of committee members.  However, we believe this is not a concern for the following reasons.  First, with just the IP address, it is non-trivial for an adversary to ``send a message'' offering bribe to a \citizen, because of carrier-grade NAT~\cite{carrierNAT} and the architecture for push notifications in smartphones ; the existing channel from a malicious \politician\ to the \citizen\ cannot be misused for this, as an untampered \blockene\ app on the \citizen\ will ignore any spurious traffic on that channel.  Second, as the committee is randomly chosen every block, pull-based bribing where the \citizens\  (who know of their selection up to 10 blocks in advance - \S~\ref{sec:comm-sel-main}) pro-actively reach out to the adversary cannot happen, as that would imply violation of the honesty assumption on \citizens, \ie, greater than 70\% being honest.}

\noindent{\bf Sybil Attack by \citizens: }
Given the lightweight cost of participation, \blockene\ needs to ensure that an adversary cannot get  disproportionate share of voting by spinning up several virtual nodes (\ie, Sybil attacks~\cite{Sybil}). A common way of addressing  Sybil attacks is {\em Proof-of-work} which is resource-intensive and does not fit the goals of \blockene; another alternative is {\em Proof-of-stake~\cite{AlgoRand}} where a participant's voting power is proportional to the amount of ``stake'' (money) on the blockchain, but it is specific to cryptocurrencies.   

In \blockene, we protect against Sybil attack by exploiting the trusted hardware (TEE) available in smartphones~\cite{android-tee, apple-tee}, and ensuring that a smartphone can have at most one  identity on the blockchain.  Thus, \blockene\ imposes an economic cost to participation, \ie, the cost of a smartphone; this is sunk cost {\em already incurred} in owning the smartphone, but protects against Sybil as each identity is a unique smartphone.  
 
In particular, each TEE has a unique public key that is certified by the platform (Android/iOS) vendor.  
The TEE can \secondrevise{certify an EdDSA public-private keypair generated by an app};  
this generated public key serves as the identity on \blockene.  The global state of \blockene\ tracks the set of valid public keys, along with the public key/certificate of the TEE that authorized it.   When a transaction for adding a new member is proposed, \blockene\ looks up the TEE public key to see if that TEE (\ie, the same smartphone) already has an identity in \blockene; if yes, it rejects the transaction\footnote{We can also support replacing the old identity with the new one for the same TEE with appropriate bookkeeping.}.   Thus, every \citizen\ on \blockene\ is tied to a unique smartphone, making it economically infeasible/unattractive for a single entity to get large participation on \blockene.

\revision{Note that \blockene\ only assumes that every certificate signed by Google/Apple for a TEE public-key corresponds to a unique smartphone.  
It does not depend on the security of an individual TEE (unlike running the blockchain consensus inside TEE, e.g. SGX~\cite{sgx-blockchain}, that opens up side-channel attacks compromising integrity and security).   As a result, the TEE identity can be replaced/combined with other unique identities.  In India, one-way-hash of Aadhaar-ID~\cite{Aadhaar,aadhaar1} (digitally verifiable, biometric-deduped, 1.2 billion-reach) can be used.  Other de-duped IDs (e.g.SSN) augmented with digital verifiability can also be used.}

\vspace{-0.1in}
\subsubsection{Attack vector of \politicians}

Dealing with 80\% dishonesty among the politicians is one of the main technical challenges in the design of \blockene. Malicious behavior by \politicians\ falls under two kinds: detectable and covert.  Detectable maliciousness where there is a succinct proof of lying, can be used to improve performance by blacklisting.
For example, if a \politician\ is supposed to only send one group of transactions in a round, but there are two versions signed by the same \politician, it is detectable with proof.  %
Covert maliciousness is harder to handle, and is the focus of our techniques.  
We list broad (non-exhaustive) classes of covert attacks a \politician\ can employ.

\noindent {\bf Staleness Attack: }  When a \citizen\ node asks the \politician\ for some state (\eg, the latest committed block), the \politician\ could return a stale block.  Such a response would appear to be valid because the old block would also have been signed by a quorum of \citizens\ (\S~\ref{subsec-latest-block}).

\noindent {\bf Split-View Attack: } A \politician\ can  respond selectively to some \citizens\ and not to others, causing a split in the world-view seen by honest \citizens.  Worse, a \politician\ can respond with two different values to different subsets of \citizens.  In a coordinated split-view attack, the malicious \politicians\ could only gossip among themselves, so that no honest \politician\ has a certain data.  Malicious \politicians\ can then selectively relay this data only to some \citizens\ (\eg, \S~\ref{subsec-commitments}).

\noindent {\bf Drop Attack: }  A malicious \politician\ may drop data written by a \citizen\ without committing it or gossiping it to other \politicians.  Similarly on a read, the \politician\ may choose to not respond, even though the \politician\ has the data (\S~\ref{subsec-safe-sample}).

\noindent{\bf Denial-of-Service Attack:} As \politicians\ are powerful servers typically hosted in the cloud, we assume that honest \politicians\ employ standard DoS protection that public clouds offer~\cite{AzureDosProtection,AwsDosProtection}.  For \citizens, most ISPs employ carrier-grade NAT to handle the explosion of IP addresses on mobile phones~\cite{carrierNAT}, which also provides DoS protection.  Malicious \politicians\ can make our gossip protocol more expensive by asking for more data than they need (\S~\ref{sec:forced-truth-gossip-main}).

\noindent{\bf Sibyl Attack:} An adversary could try pushing the dishonesty fraction of \politicians\ beyond 80\% by spinning up several nodes.  However, given the small number (say 200), we envision that \politician\ nodes would have an out-of-band registration mechanism (\eg, mapping them to real entities, say one per Fortune-500 company) - robust because only 20\% of them need to be honest (unlike \citizens).

\blockene\ protects against both detectable and covert maliciousness of the \politicians\  including the attacks listed above.

%% file: design.tex
\vspace{-0.1in}
\section{Design}
\label{sec-design}

In this section, we present in more detail how \citizens\ and \politicians\ coordinate on the key steps in \blockene.

\vspace{-0.1in}
\subsection{System Configuration}

We first outline the system configuration for \blockene.  \citizens\ in \blockene\ run on a smartphone, so we assume that their network bandwidth is low, \ie, 1 MB/s.  We choose a block size of 9MB (to amortize fixed cost per block), containing about 90k transactions (\textasciitilde 100 bytes each including a 64-byte signature).  We assume a network bandwidth of 40 MB/s between \politicians\ (representative of bandwidth in the cloud, \eg, between an Azure and a Google Cloud VM across east-US and west-US).  We choose the number of \politicians\ as 200.  The work done per block only depends on committee size, so the system scales to millions of \citizens.

Transaction originators submit signed transactions to a safe sample or to all \politicians, continuously in the background.  Transactions can modify keys that the originator has access to.  Transactions from the same originator can depend on each other; we preserve their order by tracking a per-originator {\em nonce} in the global state.  In this paper, (without loss of generality) each transaction accesses three keys (debits one key and credits another, third key is nonce).  \politicians\ gossip transactions among each other.

\vspace{-0.1in}
\subsection{Selecting Committee of \citizens}
\label{sec:comm-sel-main}

The committee of citizens for validating and signing each block is chosen on the basis of a {\em VRF (Verifiable Random Function)}~\cite{vrf}, inspired by Algorand~\cite{AlgoRand} but with one key modification.    Algorand requires each participant to check in each round whether it is chosen in the committee.  A \citizen\ on a mobile phone cannot afford to do such frequent checks because waking up the phone every round and communicating would cause significant battery drain.  Therefore, instead of computing the VRF on the hash of the previous block ($N-1$), \blockene\ uses the hash of block $N-10$, thus allowing a \citizen\ to wake up once every 10 blocks. Note that this modification still preserves the security guarantee required from VRFs in our threat model.   Specifically, for a citizen, the VRF for block $N$ is calculated as $\hash(\mysign_{\sk}(\hash(\myBlock_{N-10}) || N))$
where $\sk$ is private key known to the citizen. \footnote{\secondrevise{We use EdDSA signatures. ECDSA uses random number which the adversary can exploit to brute-force itself into the committee.}}
A \citizen\ is in the committee if the VRF has $0$'s in the last $k$ bits (hence a \citizen\ is part of a committee with probability $2^{-k}$; $k$ can be set appropriately).
Only the concerned \citizen\ can generate the VRF as it requires its private key, but anyone can verify its validity based on the public key given the signature.

\noindent{\bf Committee size: }
The size of the committee needs to balance performance and security. A small committee is good for performance, but for security of consensus protocol, we require that in any committee, at least $2/3$ \citizens\ are honest.
As our committee selection is probabilistic, by the Chernoff bound~\cite{mitzupfal}, this security requirement cannot be met for very small committee size even if we have $2/3$ honest \citizens\ overall.
Committee size increases with the fraction of dishonest \citizens.  
\revision{
We calibrate this tradeoff to obtain an expected committee size of $2000$ with a citizen dishonesty threshold of $25\%$. While these computations are described in detail in Appendix~\ref{app:comm-sel}, we provide an overview below.}

\revision{
\noindent{\bf Proof overview: }
We prove several properties about the committee for a block.  We call a \citizen\ that participates in a committee as {\em good} if the \citizen\ is honest and speaks to at least 1 honest \politician\ through $\fanout$ fan-out read/write. Otherwise, we say that the \citizen\ is {\em bad}. For a configuration with 25\% corrupt \citizens, 80\% corrupt \politicians, and $\fanout = 25$,  we show that our committee satisfies the following properties: size of all committees lies in the range $[1700 . . 2300]$ (Lemma \ref{lemma:commsize}), every committee has at least 1137 good citizens (Lemma \ref{lemma:goodsize}), every committee has at least a 2/3 fraction of good citizens (Lemma \ref{lemma:gap}), and no committee has more than 772 bad citizens (Lemma \ref{lemma:maxmaliciousnodes}).
}

\vspace{-0.1in}
\subsection{Fork-proof Structural Validation}
\label{subsec-latest-block}

\blockene\ is designed to prevent forks from occurring.  To enable this, each \citizen\ periodically verifies the {\em structural integrity} of the blockchain to enforce that the chain of hashes and VRFs are consistent and to prevent forks.

\noindent{\bf Track local state: } Each \citizen\ {\em locally} remembers the block number $N$ until which the \citizen\ validated the structural integrity of the blockchain, and the hashes of blocks $N$ to $N-9$. In addition, a \citizen\ stores an up to date list of  public keys of other valid \citizens.   The total storage size is <100MB for 1 million \citizens.

\noindent{\bf Chained ID sub-blocks: } To enable \citizens\ to efficiently update local state, the public keys of new users added as part of each block $B$, are tracked in an {\em ID sub-block (SB)} within $B$.  SBs are chained together by embedding Hash($SB_{i-1}$) within $SB_i$.  To aid cheap verification, committee members sign Hash(Hash($B_i$), Hash($SB_i$), GlobalStateRoot($B_i$)).

\noindent{\bf Incremental Validation: }  Roughly every 10 blocks (12-15 mins), each \citizen\ performs a \texttt{getLedger} call to validate the incremental structural integrity (\ie, from last validation point to the latest state), and to check if it will be in the committee soon (committee for a block $N$ is a function of the hash of block $N-10$).   To find the latest block, a \citizen\ queries a safe sample of \politicians\ for the latest block number.  It picks the highest number reported by {\em any} \politician, and asks for proof, \ie, signatures of  committee of that block and the corresponding VRFs.
Thus,  if at least one \politician\ in the safe sample is honest, the \citizen\ will know the latest block hash.   If the latest block is greater than $N+10$, it first verifies block $N+10$.
Further, it refreshes its set of valid public keys by downloading the chained sub-blocks $SB_{N+1}...SB_{N+10}$ that contain new \citizens\ added in each block, verifying the integrity of $SB_i$ based on the chained hashes.

\noindent{\bf Cool-off period for new nodes:}  To prevent a (low-probability) attack where an adversary can manufacture public-private keypairs\footnote{Android TEE API does not allow directly signing with the private key of TEE; instead a keypair is certified by TEE.} to increase chances of getting higher malicious fraction for a particular block $N$, we allow a \citizen\ to be in the committee only $k (=40)$ blocks after the block in which the \citizen\ was added.  To verify this as part of VRF checks, a \citizen's local state tracks the block number of ``recently'' added \citizens.  This is similar to the ``look back parameter'' in Algorand~\cite{AlgoRand}. %

\revision{
\noindent{\bf Proof overview:}
Our \texttt{getLedger} protocol (Appendix~\ref{app:get-ledger}) is used for verifying ledger height $i+10$, given the \citizen\ $v$  has last verified height $i$, {\em without an explicit brute-force verification of signatures of all $10$ blocks}. The algorithm generalizes to verifying any height $i+j$ for $1 \leq j \leq 10$. We show (Lemma \ref{lemma:jump}) that if a {\em good} \citizen\ with a verified state for height $i$ invokes the \texttt{getLedger} protocol at round $(i+11)$ and accepts, then the \citizen's updated {\em structural state} is consistent with the blockchain up to height $(i+10)$.  Using this, we can show that honest \citizens\ can obtain the consistent structural state of the blockchain, along with all registered public keys, for every round of the protocol (Corollary \ref{corr:get-ledger}). }

\vspace{-0.1in}
\subsection {Transaction Validation}
\label{subsec-global-state}

\citizens\ perform the task of verifying signatures of transactions, checking the transaction nonce to detect replay attacks,  and verifying semantic correctness of the transaction (\eg, double spending).  However, only \politicians\ store the Merkle tree (\S~\ref{sec:bb}) of the global state; keeping a large and up to date global state in \citizens\ is unaffordable.   To validate a transaction, a \citizen\ must lookup the correct value of keys referenced therein. On commit, the \citizen\ must {\em update} the Merkle tree with new values from the transaction, and sign the new Merkle root. The challenge lies in doing so correctly given untrusted \politicians.

The Merkle root (along with block number) is signed by the committee of the previous block, so the \politician\ cannot lie about the Merkle root.   Once the \citizen\ learns the latest block number (\S~\ref{subsec-latest-block}), it learns the correct Merkle root as well.
To verify a value returned for a key, \citizen\ asks the \politician\ to send the {\em challenge path} for this key, i.e., all the sibling nodes (hashes) along the path from the leaf to the root. This enables the \citizen\ to reconstruct the Merkle path and match the root hash with the signed Merkle root.  By security of hashes, the \politician\ cannot present spurious challenge paths that verify. %
In a tree with 1 billion key-value pairs, the challenge path would contain 30 hashes.

Update of keys in the Merkle tree follows a similar protocol.  The \citizen\ could build a partial Merkle tree with the new values at the leaves, and compute the new Merkle root.  Both the read and update paths mentioned above are expensive, and we optimize them in \S~\ref{sec-optimizations}.

\subsection{Block Proposal}
Like in any blockchain, committee members can propose a new block for committing to the blockchain. %

\subsubsection{Pick winning proposer}
\label{sec:proposer}

For efficiency, we allow only a subset of committee members called {\em proposers} to actually propose a block, based on the {\em VRF} of the \citizen.  For this selection, we use an additional VRF that is based on the hash of the previous block $N-1$ (instead of $N-10$); only committee members who have the last $k'$ bits of the additional VRF set to zero can propose a block, and the winner is the one with the least VRF.  Using the previous block hash in this VRF ensures that the adversary does not know about the proposers until the last minute (similar to Algorand) thus preventing a targeted attack on the proposers. Any committee member can consistently determine the winning VRF among the proposers.  All proposers upload their block to \politicians\ and other committee members download the block of the winning proposer.

\subsubsection{Pre-declared commitments}
\label{subsec-commitments}

The upload of the proposed block by a proposer needs to be done to a safe sample of 25 \politicians.
In \blockene, as the blocks are~$\sim$9MB in size, assuming 1MB/s bandwidth at mobile nodes, this would take 225 sec.   To optimize this step, we make the transaction selection process deterministic, so that {\em any \citizen\ can reconstruct what the original proposer would have done}, without the proposer explicitly uploading the full block.   Determinism is challenging, however, because the 80\% malicious \politicians\ can send different transactions to different \citizens.  Our technique of  {\em pre-declared commitments} to transactions addresses this.

\noindent {\bf 1. Freeze Transactions.} At the start of block $N$, each \politician\ {\em freezes} the exact set of transactions it will send to \citizens\ reading from it.   It does so by creating a \tpool, which includes a set of (about 2000) transactions, and then generates a {\em commitment} which is a signed hash of the \tpool\ along with the block number%
\footnote{To reduce overlap of transactions across \tpools\ from multiple \politicians\ (which would reduce the unique transactions in the final block), transactions are deterministically partitioned across \politicians\ using a hash on transaction identifier and round number.  Given a \tpool\ and commitment, it is easy to detect/ blacklist a \politician\ that doesn't follow this.
}.
Malicious \politicians\ are forced to issue only one commitment for a given block $N$, because two signed commitments from a \politician\ %
is a proof of malicious behavior, and can be used for efficient blacklisting; \citizens\ then drop all commitments from that \politician\ in the same round.
Intuitively, with  frozen commitments, a \citizen\ proposing a block, need not upload the full block, but only a {\em digest} with the commitments that went into the block, and other \citizens\ can reconstruct that block by downloading the \tpools\ for those commitments from \politicians.

\noindent {\bf 2. Ensure that {\em enough} honest citizens {\em have} commitments.}
A malicious \politician\ can respond with its \tpool\ only to a subset of \citizens, and refuse to respond to others; thus, a \tpool\ committed in the proposed block may not be readable by all honest \citizens, thus thwarting consensus.
To address this, we perform three steps.
First, we limit the exact set of \politicians\ from whom to pull transactions for a given block to a randomly chosen set of $45$ politicians based on the hash of the block number and hash of previous block.
 Instead of reading \tpools\ from a random safe sample, a \citizen\  reads from these 45 designated \politicians\ for a block.  
Second, the \citizen\ uploads a {\em witness list} to a safe sample of \politicians; the witness list contains the list of \tpools\ the \citizen\ was able to successfully download.
The witness list of all \citizens\ gets gossiped between \politicians.
Third, the proposer reads the witness list of all other \citizens, and picks only commitments whose \tpools\ {\em were successfully downloaded by at least a threshold number of \citizens}.
This threshold is fixed to be $\badnummax + \Delta$, where $\badnummax$ is the maximum number of malicious nodes in any committee (computed to be 772, from \S\ref{app:comm-sel}, Lemma \ref{lemma:maxmaliciousnodes}), and $\Delta$ is chosen to be 350.  Intuitively, all commitments (and \tpools) that pass this condition are available with at least $\Delta$ {\em honest} \citizens.   As 20\% of \politicians\ are honest, in expectation, at least 9 out of the 45 commitments will pass this test.

\noindent{\bf 3. Ensure that {\em all} honest citizens {\em get} commitments.}
The commitments available with at least $\Delta$ honest \citizens\ now need to be propagated to all honest \citizens.
Each \citizen\, in Step~\ref{step-first-re-upload},  {\em re-uploads} $5$ random \tpools\ it has, to $1$ random \politician.
This ensures that (with high probability) each \tpool\ (including those from malicious \politicians) that belongs to at least $\Delta$ honest \citizens\ reaches at least one honest \politician\ (who then gossips it to other honest politicians). Thus, other honest \citizens\ can successfully download that \tpool\ (by querying a safe sample of politicians), preventing a {\em split-view} attack by malicious \politicians.

\noindent{\bf 4. Handle malicious proposer.} When the winner of block proposal is a malicious \citizen, it need not respect the witness list criteria, and can pick a commitment whose \tpool\ is known to very few \citizens. This attack is possible only when consensus outputs the block proposed by this malicious proposer, so we can argue that at least $1/3$ honest \citizens\ {\em had} all \tpools\ at the beginning of the consensus. To ensure that all honest \citizens\ are able to download all required \tpools, a second re-upload of randomly chosen \tpools\ happens (step~\ref{step-second-re-upload}), now including the downloaded \tpools\ from previous step.
Formal proofs capturing the guarantees provided by these re-uploads needed to prove security of our system are presented in Lemmas \ref{lemma:honest-proposer} and \ref{lemma:malicious-proposer} of \S\ref{sec:securityproofs}

\vspace{-0.1in}
\subsection{Block Commit Protocol}
\label{sec:block-main}
The main operation in a blockchain is adding a new block to the blockchain.  We list below the key steps in the process of committing block $N$.  The protocol for block $N$ starts once the previous block $N-1$ gathers a threshold number of signatures (set to $850$ in our case, \S~\ref{app:completeprotocol}) from the committee members for block $N-1$.

\begin{newenum}
\setlength\itemsep{0.2em}
\item

A new committee of \citizens\ is chosen for block $N$ (using $\hash$ of block $N-10$), denoted by $C^N$.  The \citizens\ in $C^N$  keep polling for the latest committed block number, and start the protocol once that number is $N-1$.

\item
\label{step-basic-transactions}
Each \citizen\ $C^N_i$ in $C^N$ downloads \tpools\ \& commitments from $\rho = 45$ designated \politicians\ for the block.

\item
Each $C^N_i$ uploads a signed witness list with the commitments it downloaded, to a safe sample of \politicians.

\item
\label{step-first-re-upload}
Each \citizen\ $C^N_i$ picks 5 random \tpool\ it has, and re-uploads them to 1 random \politician.

\item
\label{step-basic-proposal}
Each {\em proposer} in $C^N$ downloads all witness lists of $C^N$ from a safe sample of \politicians, and picks commitments with at least a threshold (1122) of votes (\S~\ref{subsec-commitments}).  Then, it makes a {\em block proposal} with those commitments, along with its VRF to prove proposer eligibility.

\item
\politicians\ gossip on block proposals/VRFs  and on the \tpools\ that were re-uploaded by \citizens. %

\item
Each \citizen\ $C^N_i$ tries to download missing \tpools\ in step~\ref{step-basic-transactions} from safe sample of \politicians, relying on the re-upload (Step~\ref{step-first-re-upload})  by other \citizens.

\item
\label{step-basic-local-winner}
Each $C^N_i$  reads the VRFs of all proposers in $C^N$ from a safe sample of \politicians, and picks the lowest correct VRF as the {\em local winner}.
If $C^N_i$ already has all \tpools\ in the winning proposal, it enters consensus with that set of commitments, otherwise, NULL.

\item
\label{step-second-re-upload}
Each \citizen\ $C^N_i$ performs a second re-upload of $10$ random \tpools\ it has to 1 random \politician.

\item
\label{step-basic-consensus}
\citizens\ in $C^N$ run a consensus protocol (\S~\ref{subsec-consensus}) with gossip through \politicians, where each $C^N_i$'s vote is decided in Step~\ref{step-basic-local-winner}.  At the end, all honest \citizens\ either agree on same set of commitments or an empty block.
$C^N_i$ downloads the \tpools\ missing w.r.t. the output of consensus from safe sample of \politicians.

\item
\label{step-validation}

Each \citizen\ $C^N_i$ performs transaction validation  by downloading challenge paths for all keys from \politicians\ (\S~\ref{subsec-global-state}) and drops transactions that fail validation.

\item
Based on valid transactions (Step~\ref{step-validation}), each $C^N_i$ creates a block, computes the new Merkle root of the global state using updated values of keys and signs the block hash and new Merkle root, along with block number $N$.  It uploads the block hash, new Merkle root, and this signature to a safe sample of \politicians.

\item
\label{step-basic-commit}
When more than a threshold number of signatures have accumulated for block $N$, block $N+1$ starts.
\end{newenum}

Our complete protocol description can be found in \algoref{block-commit}, \S\ref{app:completeprotocol}. 
We give an overview of various properties of \blockene, i.e., safety, liveness and fairness, in \S~\ref{sec:proofs}.

\vspace{-0.1in}
\subsubsection{Consensus Protocol}
\label{subsec-consensus}

For consensus (Step~\ref{step-basic-consensus}), we use the Byzantine Agreement (BA) algorithm for string consensus  (that is based on \cite{tc84}) which calls upon the bit consensus algorithm BBA~\cite{micaliagreement} in a black-box manner. These are the same consensus algorithms used by Algorand.
\citizens\ enter the consensus protocol with list of commitments in local winning block, as input. Two scenarios are relevant here. If the winning proposer (\ie, the one with the lowest VRF) was honest, which would happen at least two-thirds of the time, all honest \citizens\ in the committee would enter consensus with this proposal except with small probability (\lemmaref{honest-proposer}), and the protocol will terminate in 5 rounds. However, if the winning proposer was malicious, it can collude with malicious \politicians\ to partition the view of honest \citizens.  In general, the consensus protocol would take an expected 11 rounds~\cite{AlgoRand}.

%% file: optimizations.tex
\vspace{-0.1in}
\section{Optimizations}
\label{sec-optimizations}

In this section, we present two key optimizations crucial to achieving high transaction throughput in \blockene.

\vspace{-0.1in}
\subsection{Prioritized Gossip}
\label{sec:forced-truth-gossip-main}

\noindent{\bf Problem. } The guarantee we require in \blockene\ is that if one honest \politician\ has a message, {\em all} honest \politicians\ receive the message.  Because of the high fraction of dishonesty among \politicians, standard multi-hop gossip with a small number of neighbors (\eg, 10) cannot provide this guarantee, because  there is a non-trivial probability that all of them were dishonest, and drop the message.  Hence the safe thing to do is a full broadcast to all other \politicians, which is expensive; when \politicians\ need to gossip \tpools\ that were re-uploaded by \citizens\ in the committee,  each \politician\ may have up to 45 \tpools\ to gossip; with full broadcast, it would send $0.2MB*45*200 = 1.8GB$ which would take 45 seconds in the critical path (@40MB/s).

\noindent{\bf Key idea. }  We leverage the fact that messages being gossiped by the different \politicians\ have a high overlap; each \politician\ has a subset of the same 45 \tpools\, as \citizens\ pick a random  \politician\ to re-upload a subset of \tpools. Moreover, given the nature of re-upload, in expectation,  any \politician\ would be missing only a few \tpools, and honest \politicians\ wouldn't lie about  state.

\noindent {\bf 1. Handshake. } Each  \politician\ asks recipients $B_i$ which \tpools\ they already have, and send only the missing ones.   While this works with honest \politicians, the 80\% malicious ones could always lie that don't have any, to cause a higher load/latency on the system.

\noindent{\bf 2. Selfish gossip. }  As malicious \politicians\ can lie that they have no \tpools, we assign a soft-penalty to \politicians\ that miss a lot of \tpools.  Each sender \politician\ $A$ favors the peer $B$ that has the maximum number of \tpools\ that $A$ needs.  In each round, $A$ sends a \tpool\ to $B$, and receives one in return.   Given the random re-uploads by \citizens, each honest \politician\ would be missing only a small number of \tpools, and hence would get prioritized. The list of what $B$ has to offer keeps getting updated as $B$ gets \tpools\ from other peers; note that this list can only grow, not shrink.

\noindent{\bf 3. Incentivize frugal nodes.} Selfish gossip loses its ability to discriminate between honest and malicious recipients, once the sender receives all \tpools.  To address this, after getting all \tpools, the sender changes its priority function for destinations $B_i$ to be {\em the number of \tpools\ that $B_i$ claims to have}; thus honest nodes which will have large fraction of \tpools\ are favored.  Again, the list of \tpools\ that $B$ advertises can only grow, not shrink, as shrinking would mean that $B$ lied.  \revision{Further,  each honest $B_i$ requests its missing chunk from at most $k=5$ peers simultaneously; $k=1$ will be data-frugal, but incur high latency if the peer dishonestly delays response. }

\vspace{-0.1in}
\subsection{Sampling-based Merkle Tree Read/Write}
\label{subsec:opt-gs-main}

\noindent{\bf Problem.} The Merkle tree validation in Step~\ref{step-validation}  is expensive.  In a 1-billion node Merkle-tree (30-levels deep), a challenge path is 300 bytes (10-byte hashes); downloading 270K challenge paths is 81 MB (\~81 sec latency) ignoring compression.  The compute at \citizens\ is also high (total 16.2 million hash computations  for challenge path verification during read and for computing new root post update).

\noindent{\bf Key idea.}  We offload most of this work to \politicians, in a verifiable manner.  Since the Merkle tree validation is done after the conclusion of the consensus run using gossip through the \politicians, \politicians\ know the \tpools\ that are considered for constructing the block. Hence, all \citizens\ in committee and \politicians\ know the keys whose values need to be read and updated. We first discuss the optimization for reading values correctly from the Merkle tree.

\noindent{\bf 1. Get Values.} Each \citizen\ gets just the values for all 270K keys (no challenge path, 1 MB instead of 81 MB) from one \politician, and then asks a safe sample of \politicians\ whether those values were correct.  As at least one of these \politicians\ is honest, it alerts the \citizen\ to incorrect values through an {\em exception list}.   The \politician\ can ``prove'' an incorrect value by providing a challenge path from the signed Merkle root that indicates a different value for the key.

\noindent {\bf 2. Spot-checks.} If many values were wrong, the exception list would be quite large and eat into the savings.   To avoid this, \citizen\ picks a small random subset of $k' = 4500$  keys to initially spot-check using the challenge paths.   If the spot-checks pass for a sufficiently large $k'$, a \politician\ could have lied only for a small number (200) of keys (except with small probability).  Thus, the extra spot-checks bound the size of the exception list (Lemma \ref{lemma:keyspotcheck}, \S\ref{sec:samplingreadprotocol}).

\noindent {\bf 3. Exception list protocol. }  To {\em cross-verify} the values with a safe sample of \politicians, the \citizen\ deterministically puts these values into buckets (2000) and uploads the hashes of these buckets. When a \politician\ notices a mismatch for a bucket, it sends the bucket index and the correct values for all keys in that bucket. \citizen\ gets challenge paths only for keys that disagree (from first \politician).  Our spot-checks ensure that only a small number of buckets can mismatch.

\noindent{\bf Corner case.} Even after doing the above, there is a small probability ($<2^{-10}$) that a \citizen\ may obtain an incorrect value; we count such \citizen\ nodes as malicious  and account appropriately (Lemma \ref{lemma:readincorrectkeys}). The full protocol and all proofs are provided in \algoref{read-gs} of \S\ref{sec:samplingreadprotocol}.

\noindent{\bf Writes:} Updating the Merkle tree is a trickier problem.
Due to lack of old challenge paths for the all keys being updated, the \citizen\ cannot construct the root of the updated Merkle tree $T'$.
We solve this problem by making the \politicians\ compute $T'$, but now the \citizen\ must verify that the \politicians\ performed the computation correctly, i.e.,  $T'$ is consistent with the new values of updated keys and old tree $T$ for unmodified keys.  \revision{We achieve this by breaking $T'$ at a level called the \emph{frontier level} (the nodes at this level are frontier nodes). \citizens\ obtain the values of the frontier nodes of $T'$ from a safe sample of the \politicians. The \citizens\ then run a spot checking algorithm - they pick a random subset of frontier nodes and ask a \politician\ to prove the correctness of that frontier node. Next, \citizens\ create exception lists with the help of the rest of the selected \politicians. This list denotes which frontier nodes are incorrect with the \citizen. The \citizen\ then proceeds to sequentially correct the incorrect frontier nodes and then finally compute the correct root of $T'$ from the frontier nodes.}

\revision{
\noindent{\bf Proof Overview: }
In Appendix~\ref{app:gs}, we prove (in Lemma \ref{lemma:keyspotcheck}) that for a {\em good} \citizen, after successfully spot-checking only $\mu$ fraction of key-values, only (a small number of) $\tau$ values are incorrect with probability $1-\eps_1$ (here, $\mu$, $\tau$ and $\eps_1$ are appropriately chosen parameters). Moreover, these values will get corrected by processing exception lists of size at most $\tau$. Hence, a good \citizen\ gets correct values with probability $1-\eps_1$ (Corollary \ref{corr:incorrect-keys}). We pick our parameters (Lemma \ref{lemma:readincorrectkeys}) such that at most $18$ {\em good} \citizens\ will obtain incorrect values during read, and account for these $18$, by counting them as {\em bad} \citizens\ in the committee.  In the write protocol, we can show that the sizes of exception lists can be bounded (Lemma \ref{lemma:updatekeyspotcheck}) and that no more than 18 \citizens\ accept an incorrectly updated Merkle tree $T'$ (Lemma \ref{lemma:incorrect-update}), which we once again factor in to the set of bad \citizens. We additionally also show that our algorithms are between $3-18\times$ more communication efficient and between $10-66\times$ computationally faster than the naive algorithm for global state read/write.
}

%% file: proofs-overview.tex
\section{Proofs of Safety, Liveness, and Fairness}
\label{sec:proofs}

\revision{
In this section, we provide a brief overview of the proofs detailed in the appendix for the safety, liveness, and fairness guarantees of \blockene. }

\revision{
A committee round $\commround$ ends when a new block gets signed and committed by a  threshold number ($\goodthresh$), of committee members for $\commround$. $\goodthresh$ will be set to be $850$ (done taking into account maximum number of bad citizens in any committee as well as the 36 good citizens who might have read/written an incorrect global state).}

\revision{
First, we show (in Lemma \ref{lemma:honest-proposer}) that for a block, if a good \citizen\ is the winning proposer, then (except with bounded constant probability) all good \citizens\ will output the proposal of this \citizen\ as the output of the consensus protocol. In Lemma \ref{lemma:malicious-proposer}, we show that, on the contrary, if a malicious \citizen\ is the winning proposer and the consensus results in a non-null value, then all good \citizens\ will be able to download the transactions committed in the proposal. Using Lemmas \ref{lemma:readincorrectkeys} and \ref{lemma:incorrect-update} (see Proof Overview of \S~\ref{sec-optimizations}) , we then show (Lemma \ref{lemma:blockconsensus}) that at the end of the block commit protocol all, except 36, good citizens will sign the same block hash and new global state root and that the new block is consistent with the entire blockchain and global state.
}
\revision{
Now, using Lemma \ref{lemma:blockconsensus}, safety (i.e. all honest \citizens\ agree upon all committed blocks and all blocks are consistent with a correct sequence of transactions) follows via an inductive argument. Next, to argue liveness (that adversarial entities cannot indefinitely stall the system and that the empty-block probability is bounded by a small constant), we use Lemmas \ref{lemma:blockconsensus} and \ref{lemma:honest-proposer}.
}

\revision{
Additionally, we also prove bounds on throughput in Lemma \ref{lemma:throughput} (in expectation, committed blocks have a threshold number of transactions in them) and fairness in Lemma \ref{lemma:fairness} (all valid transactions will eventually be committed).}

%% file: implementation.tex
\vspace{-0.1in}
\section{Implementation}
\label{sec-implementation}

We have built a prototype of \blockene, that is spread across two components, \citizen\ nodes and \politician\ nodes. 
\vspace{-0.1in}
\subsection{\citizen\ nodes}

The \citizen\ node is implemented as an Android app on SDK v23 and has 10,200 lines of code. It is built to optimize battery use and runs as a background app, without user involvement after initial setup. The application caters to two main phases of the protocol that a Citizen participates in: passive and active. In the passive phase, a service using JobScheduler~\cite{android-jobscheduler} periodically polls \politicians\ for \texttt{getLedger} calls. In the active phase, when the \citizen\ is part of a committee, the application runs the steps of the protocol, handling failures, timeouts and retries to deal with corrupt \politicians. The implementation for the active phase uses a multi-threaded event-driven model and is built on top of EventBus to parallelize and pipeline network and compute intensive crypto tasks such as signature validation.

\vspace{-0.1in}
\subsection{\politician\ nodes}

The \politician\ node is implemented in C++ (11K lines of code).  The implementation scales to load from thousands of \citizens, and handles bursty load during gossip.
Given the state-machine nature of the protocol, we have built it on top of the convenient C-Actor-Framework~\cite{c-actor}, which is based on ``actors'' that transition the state of the \politician\ through the steps of the protocol. For instance, the BBA actor, apart from storing and serving the votes that \citizens\ submit, also reads the votes to determine the result of consensus.  Based on this, it emits an event to build the updated Merkle tree. 
 
For the global state, we have built a SparseMerkleTree (SMT), where the leaf index is deterministically computed using the SHA256 of the key. Since the tree is of bounded depth, we allow for (a small number of) collisions in the leaf node.   The challenge path of any key includes all the collisions co-located with this key, so the leaf hash can be computed.  \revision{To prevent targeted flooding of a single leaf node, we reject key additions that take a leaf node beyond a threshold, forcing the transaction originator to use a different key}.
We also implement a DeltaMerkleTree, which allows us to efficiently create an updated version of the SMT using memory proportional only to the touched keys.
 
Our gossip implementation does simple broadcast for regular messages, and runs a stateful protocol for \tpool\ gossip. We segregate these messages into different ports/queues so the bursty gossip messages are isolated from small messages (e.g., BBA votes) that are broadcast. 
\revision{To prevent malicious \citizens\ from flooding an honest \politician\ with the responsibility of gossiping their writes, we limit the set of \politicians\ for a \citizen\ to be deterministic based on its VRF. \politicians\ do not gossip messages from non-conforming \citizens.}

%% file: evaluation.tex
\section{Evaluation}
\label{sec-evaluation}

We evaluate our \blockene\ prototype under several dimensions.  The main questions we answer in our evaluation are:

\begin{tiret}
\setlength\itemsep{0.2em}
\item
What throughput and latency does \blockene\ provide?
\item
How well does \blockene\ handle malicious behaviors?
\item
Are the optimizations on Merkle tree \& gossip useful?
\item
What is the load on \citizen\ nodes (battery/data usage)?
\end{tiret}

\subsection{Experimental setup }

In our experiments, we use a setup with 2000 \citizen\ nodes and 200 \politician\ nodes.  \citizen\ nodes are 1-core VMs on Azure with a Xeon E5-2673, 2GB of RAM, and are spread across three geographic regions across WAN: 700 VMs in SouthCentralUS, 600 VMs in WestUS, and 700 VMs in EastUS.  Each \citizen\ runs an Android 7.1 image,  and is rate-limited to 1MB/s network upload and download.  \politician\ nodes run on 8-core Azure VMs with a Xeon E5-2673, 32 GB of RAM, and are spread as 100 VMs each in EastUS and WestUS.  They are rate limited to 40MB/s network bandwidth.  Given the random safe sampling, the \citizen-\politician\ communication spans across WAN regions.  Similarly, the gossip between \politicians\ happens across WAN regions.  As our committee size is 2000, every \citizen\ is in the committee for every block.  With a higher number of \citizens, say 1 million, a particular \citizen\ will be in the committee only once every 500 blocks.  Except the per-\citizen\ load, the system performance is independent of the total number of \citizens\ and is just a function of committee size, so the numbers are representative of a large setup. 

\if 0
\begin{figure}
\centering{
\includegraphics[width=3.5in, height=2.6in, trim = {1.5cm 1.6cm 1.5cm 1.5cm}, clip]{"figures/honest-throughput"}
\caption{Throughput \& latency in honest config.}
\label{fig:trans-throughput}
}
\end{figure}
\fi

\begin{figure}[t]
\centering{
\includegraphics[width=\linewidth]{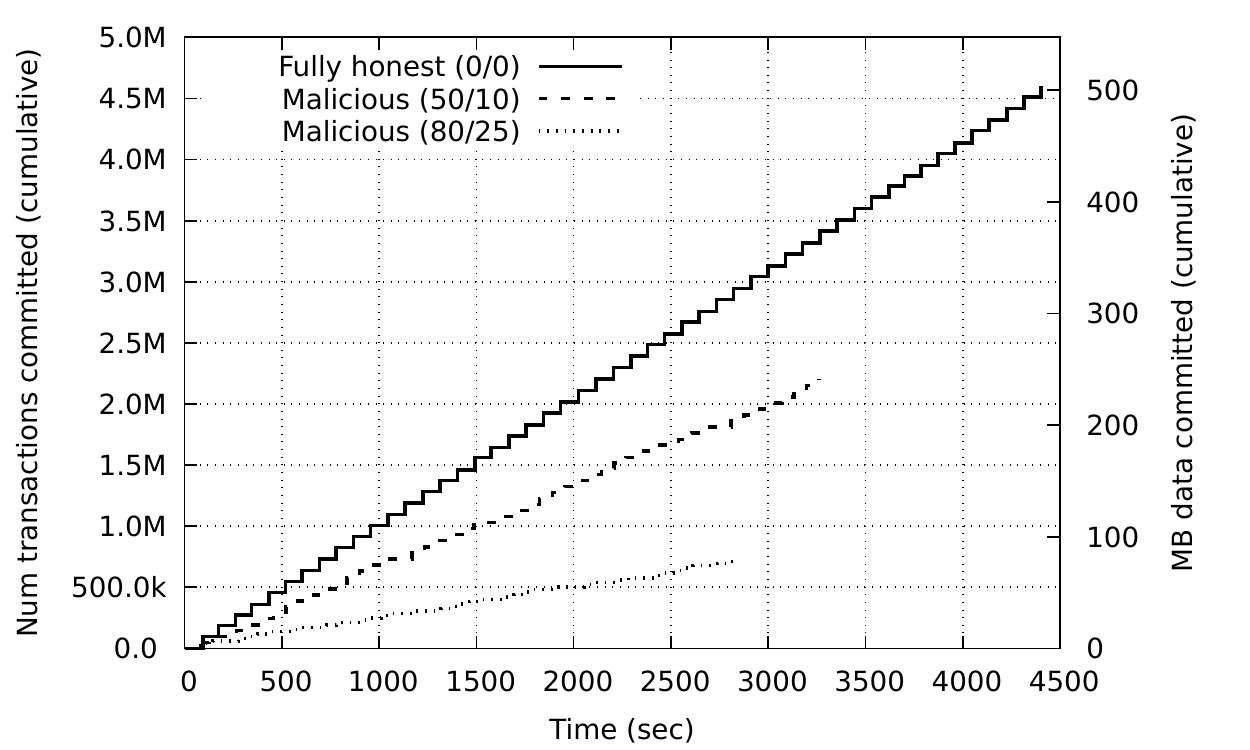}
\caption{Throughput of \blockene\ under various configs. In 50/10, 50\% \politicians\ \& 10\% \citizens\ are malicious.}
\label{fig:malicious-throughput}
}
\end{figure}

\begin{figure}[t]
\centering{
\includegraphics[width=\linewidth]{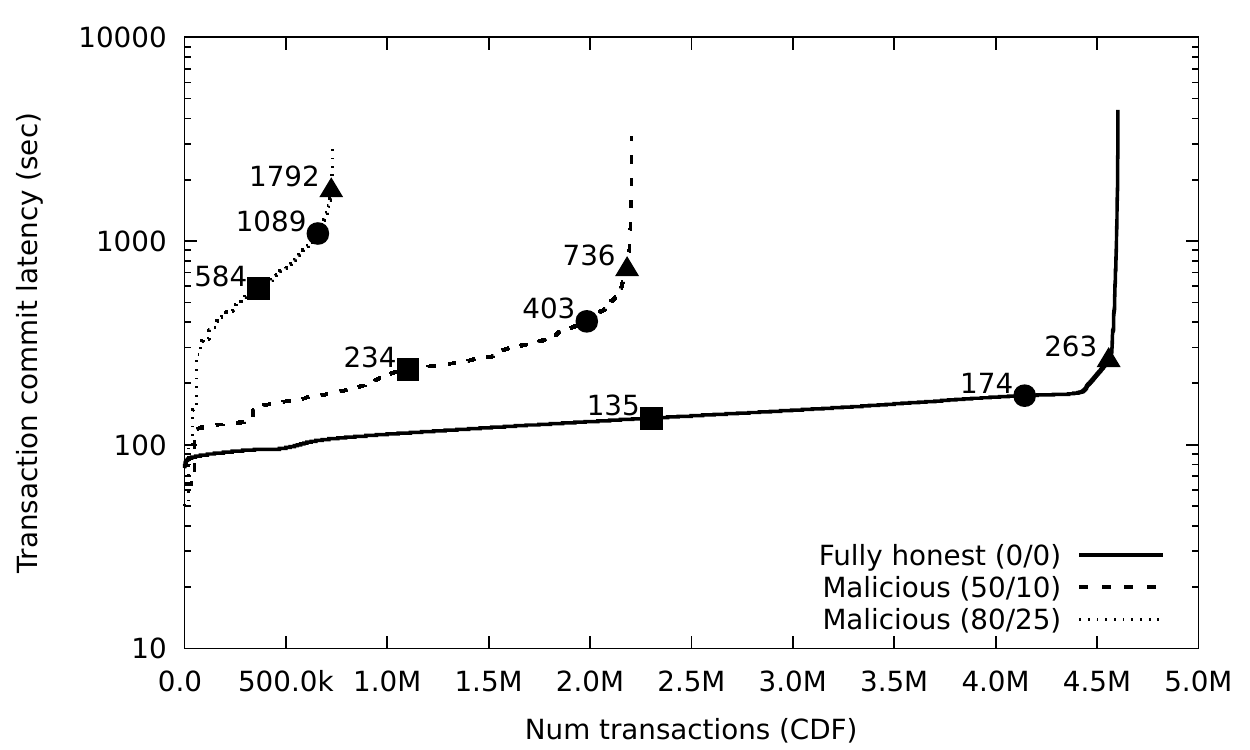}
\caption{Transaction Latency under different malicious configs. Dots show 50th, 90th, 99th percentiles.}
\label{fig:malicious-latency}
}
\end{figure}

\subsection{Transaction Throughput and Latency}

Figure~\ref{fig:malicious-throughput} shows the timeline of block commits in \blockene\ under fully honest and malicious configurations, for 50 consecutive blocks.  In the fully honest (0/0) case, 4.6 million transactions get committed in 4403 seconds, corresponding to a throughput of 1045 transactions per second, or 114 KB/s.

\begin{table}[h!]
	\begin{center}
		\footnotesize
		\begin{tabular}{ l |  r | r  | r}
			\hline
			\toprule
			\citizen\ dishonesty &  \multicolumn{3}{c}{\politician\ dishonesty} \\
			\midrule
			& 0\% & 50\% & 80\% \\
			\hline
			0\%  & 1045 & 757 & 390 \\
			10\% &  969 &  675 &  339 \\
			25\% & 813 & 553 &  257 \\
			\hline
			\bottomrule
		\end{tabular}
	\end{center}
	\vspace{-0.2in}
	\caption{{\bf Transaction throughput under malicious configs.}}
	\label{table:malicious-comprehensive}
	
\end{table}

We also evaluate \blockene\ under malicious behaviors of both \citizens\ and \politicians.   We denote our malicious configurations in the format P/C, where P is the fraction of malicious \politicians, and C is the fraction of malicious \citizens. With our choice of parameters (\eg, committee size), \blockene\ is guaranteed to ensure safety in the presence of up to 80\% malicious \politicians\ and 25\% malicious \citizens.  However, performance can be affected because of adverserial behavior.  A malicious \citizen\ in these experiments  attacks in two ways (a) force an empty block by colluding with malicious \politicians\ and proposes a block with \tpools\ that only malicious \politicians\ have.  Honest \citizens\ therefore cannot download that commitment and will vote for an empty block; (b) forces additional rounds in the BBA consensus protocol by manipulating its votes.   A malicious \politician\ attacks in two ways: (a) fails to give out transaction commitments, making a subset of the 45 \tpools\ empty, potentially causing a smaller block to be committed (b) manipulates gossip by acting as sink holes and asking for same chunks from multiple peers.
As Figure~\ref{fig:malicious-throughput} shows, \blockene\ is quite robust to a range of malicious behaviors, and gracefully degrades in performance.  With 80\% dishonest \politicians, the effective \tpools\ reduce to 9 out of 45, resulting in the block having only 18K transactions instead of 90K.  Malicious \citizens\ cause a performance hit (empty blocks + BBA rounds) when they get chosen as the proposer (\ie, highest VRF);  Table~\ref{table:malicious-comprehensive} shows the throughput under more configurations of malicious behaviors.

Figure~\ref{fig:malicious-latency} shows the CDF of transaction latencies of the system under different configurations, demonstrating fairness across transactions.  In the fully honest case (0/0), \blockene\ ensures a median latency of 135s and a 99th\%-ile latency of 263s.  Under the two malicious  configurations: 50/10 and 80/20, latencies are higher as expected.

\begin{figure}[t]
	\centering{
		\includegraphics[width=\linewidth]{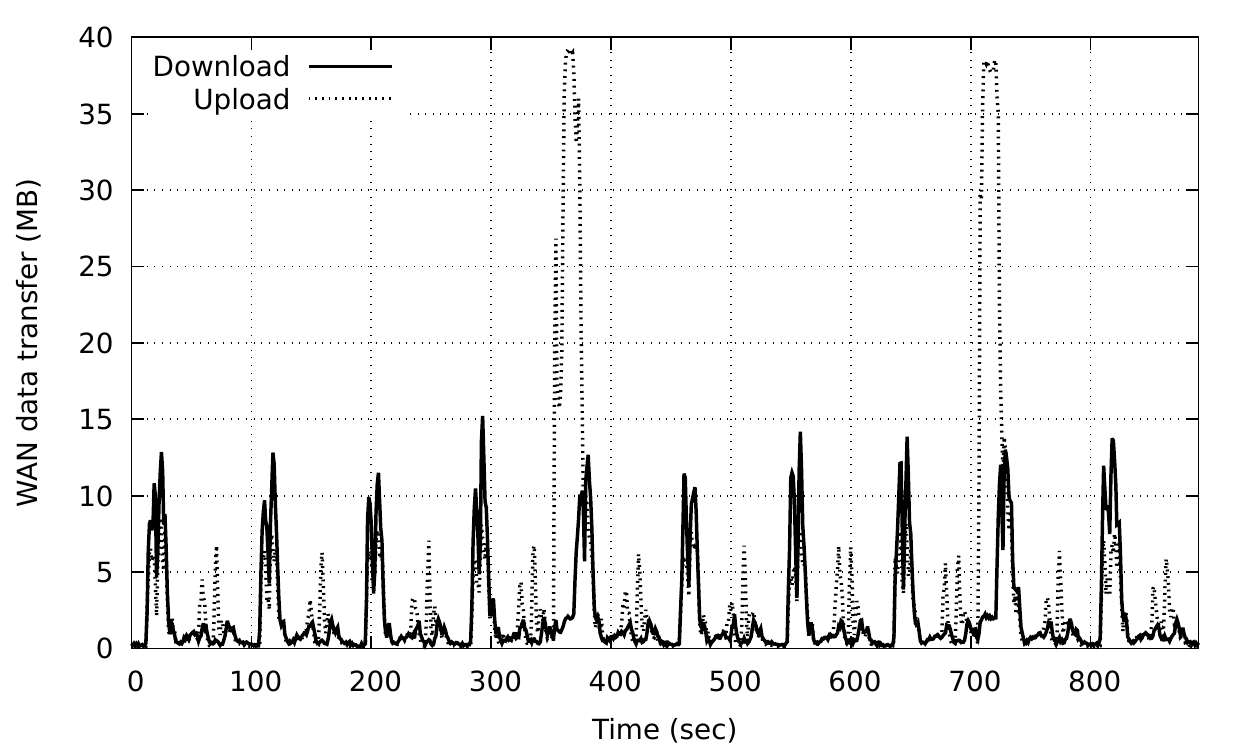}
		\caption{Network usage at a \politician\ node.}
		\label{fig:politician-network}
	}
\end{figure}

\subsection{Timeline of \citizens\ and \politicians}

Figure~\ref{fig:politician-network} shows the network load at a typical \politician\ node during 10 blocks (each of the repetitive patterns is a block).  The two large spikes in uploaded data correspond to rounds where this \politician\ was one of the 45 chosen to provide \tpools.  For each block, there are two small spikes of transmitted data; the first spike corresponds to gossip of \tpools\ through prioritized gossip, and the second spike is due to gossip of votes from \citizens\ in the BBA consensus.  

We also show the breakup of the 89-sec block latency by plotting the time taken in \citizen\ nodes during a typical block.  Figure~\ref{fig:citizen-spread} shows the progress of the 2000 \citizen\ nodes during one of the blocks, separating out the key phases of the protocol; the bulk of the time goes in the transaction validation phase, and in fetching \tpools\ from \politicians. 

\begin{figure}[t]
	\centering{
		\includegraphics[width=\linewidth]{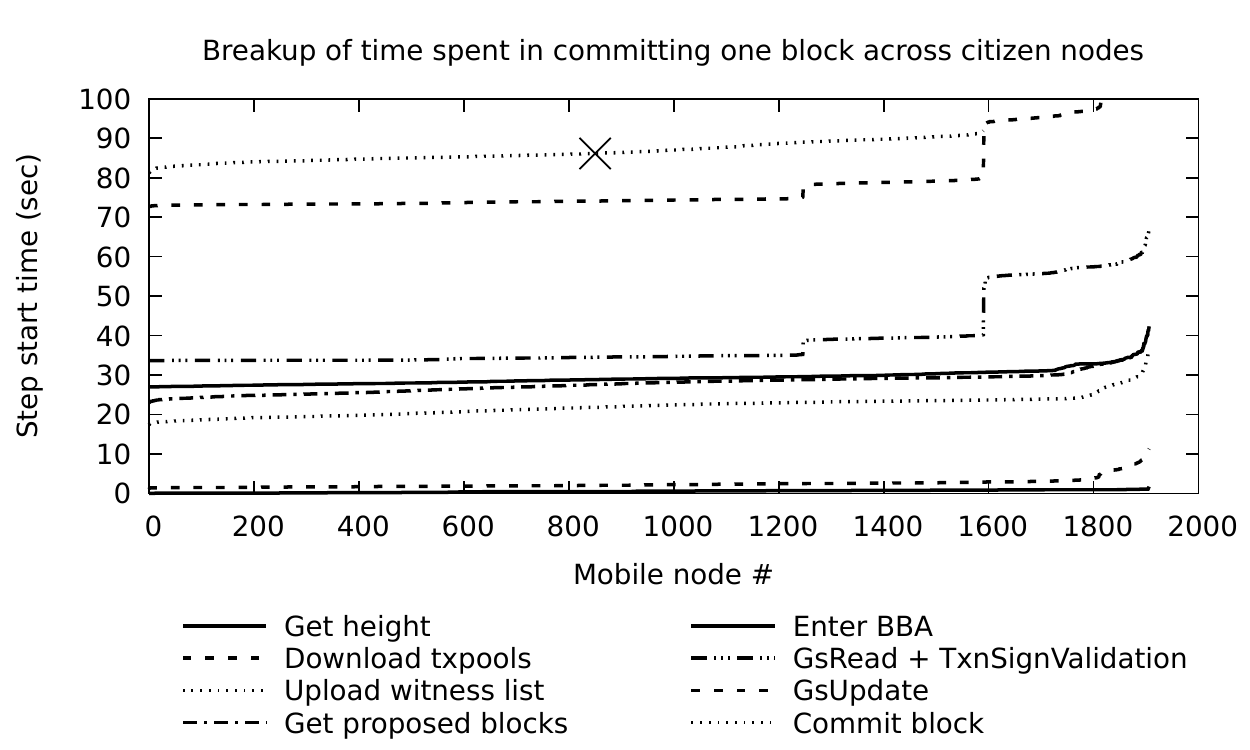}
		\caption{Breakup of time spent at \citizen\ nodes for a single block commit. \revision{Cross indicates block commit.}}
		\label{fig:citizen-spread}
	}
\end{figure}

\subsection{Impact of Optimizations}

We now evaluate the prioritized gossip and the sampling-based Merkle tree optimizations.  For gossip, we consider how much upload/download each \politician\ incurs before {\em all other honest} \politicians\ get all the \tpools.  For example, in the 0/0 case, we have 10K data points (across 50 blocks and 200 \politicians\ each). Across these samples, we plot the 50th, 90th, and 99th percentiles.  The malicious strategy we model in the 80/25 case is where \revision{only the bare minimum number of honest \citizens\ have  \tpools\ of malicious \politicians\ ($\Delta$ from \S~\ref{subsec-commitments})} and all malicious \politicians\ ask for the full set of \tpools\ from all honest nodes.   As Table~\ref{table:gossip-data} shows, the network load of prioritized gossip is robust to dishonest behavior.  Even in the malicious setting, the data transmitted is quite small before all honest \politicians\ get all \tpools.

Table~\ref{table:global-state} compares the performance of our sampling based Merkle-tree reads and updates, with the simple solution of downloading challenge paths for all keys referenced in the block.  The simple solution incurs much higher network cost (the numbers are after gRPC compression), and a significant compute cost at the \citizen.  With our optimization, the network cost drops by $10.8\times$ while the CPU cost drops by nearly $31\times$, thus significantly improving transaction throughput. 

\begin{table}
	\begin{center}
		\footnotesize
		{\color{black}{
				\begin{tabular}{ l | c |  c | c | c }
					\hline
					Config  & Percentile & Upload & Download & Time \\
					& &  (MB) & (MB) & \revision{(sec)} \\
					\hline
					0/0 & 50 & 23.1  & 22.4 & 3.6 \\
					0/0 & 90 &  30.5 &  27.5 & 4.8 \\
					0/0 & 99 &  36.7 &  30.1 & 5.2 \\
					\hline
					80/25 & 50 & 35.4  & 23.8 & 3.5 \\
					80/25 & 90 &  47.6 &  27.6 & 4.1 \\
					80/25 & 99 &  53.4 &  28.9 & 4.5 \\
					
					\hline
				\end{tabular}
		}
		}
	\end{center}
	\vspace{-0.2in}
	\caption{{\bf \revision{Cost of gossip per honest \politician\ before all honest \politicians\ receive all \tpools.}} }
	\label{table:gossip-data}
\end{table}

\begin{table}
\begin{center}
\footnotesize
\begin{tabular}{ l |  r | r  | r}
\hline
Config  & Upload & Download & Compute \\
& (MB) & (MB ) & (s) \\
\hline
Naive: GS Read & 0  & 56.16 & 93.5 \\
Naive: GS Update &  0 &  0 &  93.5 \\
\hline
Optimized: GS Read & 0.55 & 1.6 &  1.0 \\
Optimized: GS Update & 0.01 &  3 &  5.88 \\
\hline
\end{tabular}
\end{center}
\vspace{-0.2in}
\caption{{\bf Performance of Global State Read \& Write.} }
\label{table:global-state}
\end{table}

\subsection{Load on \citizens}
\label{subsec-citizen-battery}

Finally, we evaluate the load at \citizen\ nodes due to running \blockene.  The two metrics of interest are battery usage and data usage.  To get these metrics, we run an actual Android phone \revision{(a OnePlus 5)} with the \citizen\ app, as part of the committee along with the 2000 committee members on VMs, and measure battery use.   After being in the committee for 5 blocks, the battery drain was \textasciitilde 3\%. The total network traffic incurred by a \citizen\ for a single block was 19.5 MB.

Now, we can extrapolate the daily cost based on the per-block cost and the number of times a single \citizen\ is expected to be in the committee.  With 1 million \citizens, a \citizen\ will participate roughly every 500 blocks, which at our block latency of  \textasciitilde 90s, translates to 2 times per day.  Thus, the expected battery use is < 2\% per day, and the data use is \textasciitilde 40MB/day.  In addition, we also measured \revision{on the same OnePlus5 that waking up the phone every 10 minutes and performing} \texttt{getLedger} costs about 0.9\% battery and 21MB data download.  \secondrevise{Waking up every 5 minutes costs 1.7\% battery and 42MB data download.} With a total of 3\% battery usage and 61MB data/day, a user running the \blockene\ app will hardly notice it running. 

%% file: conclusion.tex
\section{Conclusion}
\label{sec-conclusion}

By enabling, for the first time, a high-throughput blockchain \revision{where members perform block validation and consensus} on smartphones at negligible resource usage, \blockene\ opens up a much larger class of real-world applications to benefit from the security and decentralized nature of blockchains.  By employing a novel architecture, and several new techniques coupled with a careful security reasoning, \blockene\ is able to simultaneously provide three conflicting properties: large scale of participation, high transaction throughput, and low resource usage at member nodes.  

\label{sec-conclusion}

\section*{Acknowledgements}

\secondrevise{We thank our shepherd Nickolai Zeldovich and the anonymous reviewers for their valuable suggestions and feedback.  We also thank Ankush Jain, Sriram Rajamani, Bill Thies, Jacki O'Neill, and Rashmi.K.Y for their support.}

%% file: appendix.tex
\sloppy
\input{prelims}
\input{good}

\section*{Sub Protocols}

We now describe some building block protocols that we will use in our final blockchain protocol.

\input{get-ledger}

\input{global-state}

\input{block-commit}

%% file: prelims.tex
\section{Preliminaries}

\paragraph{Notation.} All notation that we use to describe our protocol and proofs are provided in Figure \ref{fig:notation}. %

\begin{figure*}[hbtp]
\begin{framed}
\centering

\begin{enumerate}
\item Politicians
\begin{tiret}
\item $\sn$: set of politicians.
\item $\snnum$: Number of politicians; i.e., $|\sn| = \snnum$.
\item $\corrsn$: upper bound on fraction of corrupt politicians.
\end{tiret}
\item Citizens
\begin{tiret}
\item $\mobilenum$: Number of citizens.
\item $\hmaj$: lower bound on fraction of honest citizens.
\end{tiret}
\item Committee
\begin{tiret}
\item $\commnum$: Number of citizens in a committee in a given round.

\item $\commnumlow$: Lower bound on the number of citizens in a committee in any round.
\item $\commnumhigh$: Upper bound on the number of citizens in a committee in any round.
\item $\commnummean$: Mean value of number of citizens in a committee in any round.

\item $\heavynum$: Number of block proposers in a committee in a given round.

\end{tiret}
\item Protocol Specific
\begin{tiret}
\item $\fanout$: Fan-out of read/write from citizens to politicians.
\item $\allkeys$: Set of all keys that are part of the global state. $|\allkeys| = k$.
\item $\deep$: Depth of the merkle tree used for global state, i.e., number of leaves in the tree is $2^\deep$.
\item $\hashsz$: Size of the hash function used in merkle tree in bytes.
\item $\collide$: All leaves in the Merkle tree will have $< \collide$ number of keys (except with tiny probability). %
\end{tiret}
\item Parameters
\begin{tiret}
\item Size of keys and values is 4 bytes.
\item Size of transaction UUID is 8 bytes.

\item $\secparam$: Security parameter for probabilistic events; when we say that a quantity is negligible in $\secparam$, we mean $2^{-\secparam}$ with $\secparam$ set to $30$.
\end{tiret}

\end{enumerate}
\end{framed}
\caption{Notation}
\label{fig:notation}

\end{figure*}

%% file: good.tex
\section{Committee Selection}
\label{app:comm-sel}

In our protocol, every citizen is in any particular committee with probability $p = \frac{\commnummean}{\mobilenum}$. This means that the expected size of every committee is $\commnummean$. As described in Section \ref{sec:comm-sel-main}, our cryptographic sortition mechanism of selecting committee members for round $\commround$ is performed by computing a VRF on the hash of the block at $\commround-10$, concatenated with round number $\commround$. Similar to Algorand~\cite{AlgoRand}, we require the public keys of the VRF of citizens to be added to the blockchain at least $k$ ($k = 40$) rounds before the round in which they are eligible to be in a committeee (see Section~\ref{subsec-latest-block}). Additionally, since in our attack model, an adversary cannot change the set of corrupted citizens of round $\commround$ after round $\commround-10$ (see Section~\ref{sec:threat}), it is easy to see (using a similar proof as Algorand~\cite{AlgoRand}), that the adversary, has no information about whether an honest citizen will be in the committee at round $\commround$ or not, when it chooses the citizens who will be corrupted at round $\commround$. Hence, from the adversary's perspective, every citizen is in the committee of round $\commround$ with probability $p$ (as defined above). We can now utilize this fact below in all our lemmas and proofs. Similarly, a committee member is also chosen to be a proposer with appropriate probability so that on average each round has $\heavynum$ number of block proposers (by the choice of our function to select proposers, this is hidden and uniformly random to an adversary until round $\commround-1$).

Lemma \ref{lemma:commsize} proves a lower and an upper bound on the number of citizens in any committee sampled using above probability. Let $\kl{x}{y} = x \ln \frac{x}{y} + (1-x) \ln \frac{1-x}{1-y}$ be the Kullback-Leibler divergence between Bernoulli distributed random variables with parameters $x$ and $y$ respectively.

\begin{lemma}[Bound on committee size]
\label{lemma:commsize}
Let $\eps_c > 0$. Then, except with probability $p_c := \exp\left(-\kl{p-\eps_c}{p}\mobilenum\right)$, $\commnum \geq \commnumlow := \mobilenum(p-\eps_c)$. Similarly, except with probability $p_c' := \exp\left(-\kl{p+\eps_c}{p}\mobilenum\right)$, $\commnum \leq \commnumhigh := \mobilenum(p+\eps_c)$. In particular, with $Np =2000$, we can choose $\eps_c > 0$ such that $1700 \leq \commnum \leq 2300$, except with probability $2^{-\secparam}$.
\end{lemma}

\begin{proof}
This follows from standard Chernoff bounds.
\end{proof}
In what follows, we will refer to the range $[1700 . . 2300]$ as the \emph{probable range} of committee size.

As we describe later, in our sub-protocols for block-commit, global state read, and global state update, a citizen begins by picking a random subset of $\fanout$ politicians that this citizen communicates with.
Next, we define the  {\em good} and the {\em bad} citizens and prove some bounds on their sizes.

\begin{definition}
\label{defn:good-citizen}
We call a citizen that participates in a committee as {\em good} if the following two properties hold:

\begin{enumerate}

\item The citizen is honest;
\item The citizen speaks with at least 1 honest politician (through $m$ fan-out read/write).

\end{enumerate}

Otherwise, we say that a citizen that participates in a committee is {\em bad}.

Let $\goodnum, \badnum$ denote the number of good and bad citizens in a committee of size $n$, respectively. We will denote by $\goodnum^*$ a lower bound on the number of good citizens in any committee with size in the probable range and by $\badnummax$ an upper bound on the number of malicious citizens similarly in any committee with size in the probable range.

\end{definition}

We shall set $\hmaj = 0.75, %
\corrsn = 0.8, \fanout = 25$. %
Below, \lemmaref{goodsize} proves a lower bound on the number of good citizens in any committee, i.e., $\goodnum^*$. Also, define $\gap = \goodnum - 2\badnum$. \lemmaref{gap} proves a lower bound on $\gap$. Finally, \lemmaref{maxmaliciousnodes} proves an upper bound on number of malicious nodes, i.e., $\badnummax$.

\begin{lemma}[Lower bound on good citizens]
\label{lemma:goodsize}

Let $\eps_f, \eps_g>0$. Then, except with probability $p_{\textsf{gf}} := \exp\left(-\kl{\hmaj-\eps_g}{\hmaj}\commnum\right)+\exp\left(-\kl{\corrsn^\fanout+\eps_f}{\corrsn^\fanout}\commnum\right)$, $\goodnum \geq \left(1-\corrsn^\fanout-\eps_f\right)\left(\hmaj  - \eps_g\right)\commnum$.

In particular, we have $\goodnum^* \geq 1137$, except with negligible probability
\end{lemma}

\begin{proof}
When the committee size is $\commnum$, the expected number of honest citizens is $\hmaj  \commnum$. Hence,  except with probability $\exp\left(-\kl{\hmaj-\eps_g}{\hmaj}\commnum\right)$, there are at least $\commnum_w :=(\hmaj -\eps_g)\commnum$ honest citizens. Now, out of these citizens, the expected number that speak with at least $1$ honest politician is $(1-\corrsn^\fanout)\commnum_w$ and hence, except with probability $\exp\left(-\kl{\corrsn^\fanout+\eps_f}{\corrsn^\fanout}\commnum_w\right)$, a $(1-\corrsn^\fanout-\eps_f)$ fraction of these citizens will speak with at least 1 honest politician. Combined, this gives us the first part of the lemma.

For the second part, Lemma~\ref{lemma:commsize} shows that  $\commnum \geq \commnumlow$, except with probability $p_c$. Combined with the first part, we conclude that $\goodnum \geq \left(1-\corrsn^\fanout-\eps_f\right)\left(\hmaj  - \eps_g\right)\commnumlow$, except with probability $p_{\textsf{gf}}^* := \exp\left(-\kl{\hmaj-\eps_g}{\hmaj}\commnumlow\right) + \exp\left(-\kl{\corrsn^\fanout+\eps_f}{\corrsn^\fanout}\commnumlow\right)+p_c$. The second part now follows since  $\commnumlow = 1700$ for the probable range of $n$ from Lemma~\ref{lemma:commsize} and we can choose small enough $\eps_f, \eps_g>0$ to make $p_{\textsf{gf}}^*$ negligible.

\end{proof}

\begin{lemma}[Lower bound on $\gap$]
\label{lemma:gap}
Let $\eps_m > 0$.
Except with probability $p_{\textsf{gap}} := \exp\left(-\kl{1-\hmaj+\eps_m}{1-\hmaj}\commnum\right)+2p_{\textsf{gf}}$, $\gap \geq \left(\hmaj(1-3\corrsn^\fanout-3\eps_f)+2\hmaj-2-\eps_g-2\eps_m+\eps_f\eps_g\right)\commnum$.
In particular,
$\gap \geq 1$, except with negligible probability.

\end{lemma}

\begin{proof}
Let $\commnum$ be the committee size. Now, from Lemma \ref{lemma:goodsize}, we have that $\goodnum \geq \left(1-\corrsn^\fanout-\eps_f\right)\left(\hmaj  - \eps_g\right)\commnum$, except with probability $p_{\textsf{gf}}$. Now, a citizen in the committee is bad if it is either malicious or if the citizen spoke only to malicious politicians through the fanout read/write. Except with probability $\exp\left(-\kl{1-\hmaj+\eps_m}{1-\hmaj}\commnum\right)$, the number of malicious citizens in the committee is $\leq (1-\hmaj+\eps_m)\commnum$. To additionally bound the number of citizens that speak only to malicious politicians through the fanout read/write, note that this value is $\leq \left((\hmaj  + \eps_g)(\corrsn^\fanout+\eps_f)\right)\commnum$, except with probability $\exp\left(-\kl{\hmaj+\eps_g}{\hmaj}\commnum\right)+\exp\left(-\kl{\corrsn^\fanout+\eps_f}{\corrsn^\fanout}\commnum\right)$, which is $\leq p_{\textsf{gf}}$. Hence, $\badnum \leq (1-\hmaj+\eps_m+(\hmaj+\eps_g)(\corrsn^\fanout+\eps_f))\commnum$, except with probability $\exp\left(-\kl{1-\hmaj+\eps_m}{1-\hmaj}\commnum\right)+p_{\textsf{gf}}$. This gives us that $\gap = \goodnum-2\badnum \geq \left(\left(1-\corrsn^\fanout-\eps_f\right)\left(\hmaj  - \eps_g\right)-2(1-\hmaj+\eps_m+(\hmaj+\eps_g)(\corrsn^\fanout+\eps_f))\right)\commnum = \left(\hmaj(1-3\corrsn^\fanout-3\eps_f)+2\hmaj-2-\eps_g-2\eps_m+\eps_f\eps_g\right)\commnum$, except with probability $p_{\textsf{gap}}$. The lemma follows from Lemmas~\ref{lemma:commsize} and \ref{lemma:goodsize} and choosing a small enough $\eps_m > 0$ to make $p_{\textsf{gap}}$ negligible.

\end{proof}

\begin{corollary}\label{corr:committee-bounds} Above lemma shows that $\gap$ increases monotonically with committee size $\commnum$. Moreover, for any committee size, since $\gap \geq 1$, it holds that $\goodnum > 2\commnum/3$ and $\badnum < \commnum/3$.
\end{corollary}

\begin{lemma}[Upper bound on malicious citizens]\label{lemma:maxmaliciousnodes}
Let $\badnummax$ denote the maximum number of bad nodes in any committee. Then, except with negligible probability, $\badnummax \leq 772$.
\end{lemma}
\begin{proof}
In the proof of Lemma~\ref{lemma:gap}, we already saw that $\badnum \leq (1-\hmaj+\eps_m+(\hmaj+\eps_g)(\corrsn^\fanout+\eps_f))\commnum$, except with probability $\exp\left(-\kl{1-\hmaj+\eps_m}{1-\hmaj}\commnum\right)+p_{\textsf{gf}}$. Lemma~\ref{lemma:commsize} shows us that $n \leq 2300$ for the probable range of committee size except with negligible probability. Combining these two estimates and choosing suitably small $\eps$'s, we conclude the claimed upper bound on $\badnummax$. 
\end{proof}

In Table~\ref{table:comm-size} we show how the committee size varies with some candidate choice of corruption thresholds and size of random fan-out read/writes at the \politicians.

\begin{table*}
\begin{center}
\begin{tabular}{ |c |  c  | c|}
\hline
Corruption threshold  & Corruption threshold & Mean Committee \\
of \citizens & of \politicians & size \\
\hline
0.2 & 0.8 &  820 \\ \hline
0.2 & 0.75 & 700 \\ \hline
0.25 & 0.8  & 2000 \\ \hline
0.25 & 0.75  & 1850 \\ \hline
0.3 & 0.8 & 14000 \\ \hline
0.3 & 0.75 & 12000 \\ \hline
\end{tabular}
\end{center}
\vspace{-0.2in}
\caption{{\bf Average  Committee size with varying corruption thresholds; $\fanout = 25$.}}
\label{table:comm-size}
\end{table*}

%% file: get-ledger.tex
\newcommand{\ewnp}{except with negligible probability}
\section{Get Ledger Protocol}
\label{app:get-ledger}

We will denote by $\tk$ the public key of the TEE of a citizen %
and by $\vk$ the signature verification key of the same citizen on the blockchain. The signature key pair is generated within the TEE and $\vk$'s certification by $\tk$ is verified at the time the citizen is added to the blockchain with the \texttt{addNewNode} transaction.

In Blockene, we have a chain of blocks of transactions and an implicit chain of valid identities (public keys and certificates on them) and previous block's hashes. Denote by $\subblock_i$ the \emph{sub-block} in $\block_i$ containing new citizen  identities added in $\block_i$, $\hash(\block_{i-1})$,  and $\hash(\subblock_{i-1})$. Here, an identity is defined by a tuple $\left(\tk, \cert(\tk), \vk, \cert(\vk)\right)$, where $\cert(k)$ is a suitable certificate for public/verification key $k$. Also, a committee member of block $\block_i$ signs $\left(\hash(\block_i), \GSRoot_i, \hash(\subblock_i)\right)$, where $\GSRoot_i$ is the root of global state Merkle tree after $\block_i$, and $\subblock_i$ is the  sub-block in $\block_i$ defined as above.

Most of our sub-protocols and main blockchain protocol assume that a good citizen can reliably learn the height of the blockchain and the set of all valid public keys on the blockchain at that height. Here, we describe a protocol \texttt{getLedger} that realizes this assumption.

Define
\begin{align}
\label{eq:idpk}
\idpk_i & :=  \{(\tk,\vk, i) : \text{identity $(\tk, \vk)$ is added by $\subblock_i$}\}, \\
\label{eq:gspk}
\gspk_i & :=  \idpk_0 \; \cup \: \cdots \: \cup \; \idpk_i, \\
\label{eq:hash-chain}
\hc_i & := \{\hash(\block_{i-9}), \ldots, \hash(\block_{i-1})\}.
\end{align}
So, $\gspk_i$ is the set of all (TEE public key, verification key) pairs added to Blockene by the time $\block_i$ is committed, including the block number when the citizen with that pair is added and $\hc_i$ is the set of hashes of ``skipped'' blocks between successive calls to \texttt{getLedger} described below.

Our citizen nodes are stateful and maintain a local state of last verified height of blockchain and identities. If the last verified height of the blockchain by a citizen $v$ is $i$,  then \emph{local state} $\localstate(v,i)$ of that citizen is defined as follows.
\begin{align}
\label{eq:localstate}
\localstate(v, i) := \left(\hash(\block_{i}), \GSRoot_i, \hash(\subblock_i),  \gspk_i, \hc_i \right).
\end{align}

\emph{Committee Membership}: As described in Section~\ref{sec:comm-sel-main}, an identity $(\tk_j, \vk_j) \in \gspk_{i-50}$ is in committee for block $\block_i$ if $\sortition(\sk_j, \hash(\block_{i-10}), i)$ outputs $1$ and this can be verified without knowing $sk_j$ given $vk_j$ .

We describe in Algorithm~\ref{algo:getLedger} below the  \texttt{getLedger} protocol for verifying ledger height $i+10$, given the verifier $v$  has last verified height $i$,  without an explicit brute-force verification of signatures of all $10$ blocks. The algorithm naturally generalizes to verifying any height $i+j$ for $1 \leq j \leq 10$. The jump of ``10'' here is due to the way we define the cryptographic sortition function to determine the committee membership.

\begin{algorithm*}
\caption{\texttt{getLedger} by a citizen node $v$}
\label{algo:getLedger}
\begin{algorithmic}

\Require Citizen $v$ has verified ledger height $i$ and stores local state $\localstate(v,i)$ as in \eqref{eq:localstate}. Citizen $v$ talks to a politician node claiming to have a correct copy of the blockchain up to block $\block_{i+10}$.
\Ensure Citizen $v$ either rejects or accepts and updates local state to $\localstate(v, i+10)$.
\\

\begin{enumerate}
\item Download $\hash(\block_{i+10})$, global state root $\GSRoot_{i+10}$, sub-block hash $\hash(\subblock_{i+10})$, $\goodthresh$ block commit signatures ($\goodthresh = 850$ as set in Section~\ref{app:completeprotocol}), and membership proofs for these $\goodthresh$  committee members for $\block_{i+10}$. By traversing the sub-blockchain, also download $\subblock_{i+10}, \ldots, \subblock_{i+1}$ .

\item For each of the $\goodthresh$ citizens $(\tk, \vk)$ who signed the block commit signatures for $\block_{i+10}$, the verifier $v$ checks that $(\tk, \vk) \in \gspk_{i+10 - \cooloff}$. Given $v$ has $\hash(\block_{i})$ in its local state, committee membership proofs for $\block_{i+10}$ for each such $(\tk, \vk)$ can be verified.  Finally, signatures on $\left(\hash(\block_{i+10}), \mathsf{GSRoot}_{i+10}, \hash(\subblock_{i+10})\right)$ can be validated for each of these committee members.

\item The verifier $v$ checks that $\subblock_{i+j}$ contains $\hash(\subblock_{i+j-1})$, for $10 \geq j \geq 1$; note that $v$'s local state contains $\hash(\subblock_i)$. Next, $v$ extracts sets of $(\tk, \vk)$ in $\idpk_{i+j}$ from $\subblock_{i+j}$, for $1 \leq j \leq 10$. Then, $v$  checks certificates in all the tuples $\left(\tk, \cert(tk), \vk, \cert(\vk)\right)$ for $(\tk, \vk)$ added in $\idpk_{i+1}, \ldots, \idpk_{i+10}$ and accumulates them into $\gspk_{i+1}, \ldots, \gspk_{i+10}$. Finally, it verifies that $\subblock_{i+1}$ contains $\hash(\block_i)$ and also extracts $\hash(\block_{i+j-1})$ from $\subblock_{i+j}$ for $2 \leq j \leq 10$. In particular, the verifier now has $\gspk_{i+10}$ and $\hc_{i+10}$.

\item If any of the checks in the above steps fails, $v$ rejects. Otherwise, $v$ accepts and updates the new local state to $$\localstate(v,i+10) = \left(\hash(\block_{i+10}), \GSRoot_{i+10}, \hash(\subblock_{i+10}), \gspk_{i+10}, \hc_{i+10} \right).$$
\end{enumerate}

\end{algorithmic}
\end{algorithm*}

Recall (Definition~\ref{defn:good-citizen}) that a \emph{good citizen} is an honest committee member who talks to at least one honest politician.

The following lemma captures the correctness of \texttt{getLedger}.
\begin{lemma}
\label{lemma:jump}
Suppose $v$ has the correct local state for height $i$, i.e., $\localstate(v,i)$ in \eqref{eq:localstate} is consistent with the current blockchain up to height $i$.

If $v$ is a good citizen for round $(i+11)$ and calls \texttt{getLedger} (Algorithm~\ref{algo:getLedger}) with local state $\localstate(v,i)$ and accepts, then  $v$'s updated local state $\localstate(v,i+10)$ is consistent with the same blockchain up to height $i+10$, \ewnp.
\end{lemma}
\begin{proof}
Note that $v$ talks to at least one honest politician, say $P$, who holds the correct copy of the blockchain up to height $i+10$. We will argue that if $v$ accepts, then, \ewnp, $v$'s view of the blockchain is consistent with that of $P$'s  from height $i$ to height $i+10$, and hence, with the \emph{unique} blockchain up to height $i+10$ as maintained by $P$. From Theorem~\ref{theorem:safety}  proved in Section~\ref{sec:block-commit}, \ewnp, there is a unique sequence of blocks $\langle \block_{i+1}, \cdots, \block_{i+10} \rangle$ and global state root $\GSRoot_{i+10}$ that extend the blockchain up to $\block_i$ and $\GSRoot_{i}$.  Since $\localstate(v,i)$ is correct and the committee for $\block_{i+10}$ is determined by height $i$ prefix, in step (2), the committee that $v$ evaluates signatures of and the committee that blessed $\block_{i+10}$ in $P$'s copy are identical. Hence, among all politicians,  there is a unique tuple $\left(\hash(\block_{i+10}), \mathsf{GSRoot}_{i+10}, \hash(\subblock_{i+10})\right)$ in step (2) that would make $v$ accept and this must be consistent with $P$'s copy of the blockchain and global state. As argued above, Theorem~\ref{theorem:safety} implies that blocks $\block_{i+9}, \ldots, \block_{i+1}$, and hence their sub-blocks, also must be unique. By collision resistance of hash functions, $\hash(\block_{i+j})$ and $\hash(\subblock_{i+j})$ are also uniquely determined for $1 \leq j \leq 10$. It follows that  the sub-blocks $\subblock_{i+j}, 1 \leq j \leq 10$ that $v$ downloaded in step (1) using the sub-block hash chain are correct; note that $v$ checks $\subblock_{i+10}$ with committee's signatures on $\hash(\subblock_{i+10})$ and containment of $\hash(\block_i)$ in $\subblock_{i+1}$. In particular, $v$ obtains the correct values -- according to $P$'s copy of the blockchain -- for $\hash(\block_{i+9}), \ldots, \hash(\block_{i+1})$ and the identity tuples $\left(\tk, \cert(tk), \vk, \cert(\vk)\right)$ added in these sub-blocks.  Finally, in step (3), $v$ would validate the same identity certificates as the respective committees would have verified before adding them to these sub-blocks. Hence, if $v$ accepts, it would correctly compute $\idpk_{i+j}$, $\gspk_{i+j}$, and $HC_{i+j}$ for $ 1 \leq j \leq 10$.
\end{proof}

An easy inductive argument proves the following corollary.
\begin{corollary}
\label{corr:get-ledger}
Let $v$ be a good citizen for $\block_{i+1}$ for $i \geq \cooloff$. Then, \ewnp,  $v$ can acquire the correct local state $\localstate(v, i)$ for ledger height $i$  as in \eqref{eq:localstate}.
In particular,
\begin{enumerate}[label=(\roman*)]
\item $v$ obtains the correct global state root $\GSRoot_i$, and
\item $v$ downloads $\gspk_i$ into its local state, i.e., knows all the valid citizen identities $(tk,vk)$ registered up to $\block_i$.
\end{enumerate}
\end{corollary}

\paragraph{Communication cost of \texttt{getLedger}}: The verifier $v$ downloads the following from a politician. $850$ signatures of 64-bytes each on the triple $(\hash(\block_{i}), \GSRoot_i, \hash(\subblock_i)$ (less than 0.1KB) and for each of the $850$ committee members that signed, we need their verification keys, VRF values, and VRF proofs. This gives a total of about 136 KB. In addition, $v$ also downloads 10 subblocks and the id's and hashes in them. For concreteness. we assume each block allows one identity to be added. An identity $\left(\tk, \cert(tk), \vk, \cert(\vk)\right)$ along with the certificates is about 3.5KB. Heence a sub-block $\subblock_j$ with its id's,  $\hash(\block_{j-1})$,  and $\hash(\subblock_{j-1})$ is about 3.6KB. Hence, the 10 sub-blocks cost about 36KB. So, from each politician, $v$ downloads a total of about  172KB. Since a politician is blacklistable if the verification fails, in the most common case, $v$ needs to download only from one politician making \texttt{getLedger}'s download cost is less than 0.18MB in the common case. Since a citizen makes about 150 calls to \texttt{getLedger} in a day, the total download cost for a citizen node is roughly 27MB per day.

%% file: global-state.tex
\section{Merkle tree of Global State}
\label{app:gs}
All the politicians maintain the up-to-date global state in the form of a Merkle tree whose root is signed by the committee in each round. 
The depth of this tree is $\deep =30$. Keys are mapped to random leaves and we assume that at most $\collide=10$ keys can get mapped to a single leaf.
In this section, we describe formally our optimized protocols for global state read and update by the light-weight citizens for the keys referenced in a tentative block of transactions. 
In our overall protocol (as described in Appendix~\ref{sec:block-commit}), these steps are executed after the completion of the consensus protocol. Hence, as proved formally in Appendix~\ref{sec:securityproofs}, we can safely assume that all good citizens as well as politician nodes have access to the transaction pools referred by the output of the consensus. This ensures that both the citizens and the  politicians know the keys that are being referenced in the proposed block. First, we describe our protocol for global state read that is used by citizen nodes to learn the correct values corresponding to these keys. 
Next, the citizen will apply relevant transactions on these key-values to obtain updated values for each of these keys. Finally, we describe our protocol for global state update that is used by citizens to obtain the merkle root of the tree that has updated values of each of these keys. 
For both of our protocols we prove that all good citizens except $36$ will learn the correct values during global state read and correct new global state Merkle root in global state update except with negligible probability. We use this to define effective good citizens that would sign the correct block. %

\subsection{Sampling based Merkle Tree Read}
\label{sec:samplingreadprotocol}
Here, we describe a protocol to realize the following task. A citizen and politicians hold a set of $\keys \subset \allkeys$ and the citizen wants to read the corresponding values from the global state. Let $k = |\keys|$.
We describe our optimized protocol for reading values from global state 
in \algoref{read-gs}.

\input{algo-read-gs}
At a high level, as described in \S~\ref{subsec:opt-gs-main}, citizen first does a spot-check on small number of keys by downloading the complete challenge paths and then corrects the incorrect key-values by cross validating across $\fanout$ politicians.

First, we prove that for a good citizen, after successfully spot-checking only $\mu$ fraction of key-values in Steps~\ref{gsread:sc1}, \ref{gsread:sc2}, \ref{gsread:sc3} only (a small number of) $\tau$ values  are incorrect with probability $1-\eps_1$ (here, $\eps_1$ is a function of $\mu$ and $\tau$). Moreover, these values will get corrected by processing exception lists of size at most $\tau$ in Steps~\ref{gsread:ex1}, \ref{gsread:ex2}, \ref{gsread:ex3}. Hence, a good citizen gets correct values with probability $1-\eps_1$. More formally,

\begin{lemma}\label{lemma:keyspotcheck}
For an honest citizen, if verification in step~\ref{gsread:sc3} succeeds, then at most $\tau$ values can be incorrect in Step~\ref{gsread:sc1}, except with probability at most $\eps_1 = e^{-\mu\tau}$. 
\end{lemma}
\begin{proof}
The lemma follows from a standard probability argument. Let $t > \tau$ values be incorrect in Step~\ref{gsread:sc1}. Then the probability with which an honest citizen does not pick an incorrect value for verification is at most $\left(1-\frac{t}{k}\right)^{\mu\cdot k}\leq e^{-\mu\tau}$, for $t> \tau$. 
\end{proof}

\begin{corollary}[Upper bound on failure probability of Global State Read]
\label{corr:incorrect-keys}
A good citizen learns the correct values for all keys $\keys$, except with probability at most $\eps_1$.
\end{corollary}
\begin{proof}
It follows from the above lemma and the fact that for a good citizen at least one of $S_1, \dotsc, S_\fanout$ is honest. This honest politician node would provide the correct exception list that is used by citizen node to correct the incorrect values received in Step~\ref{gsread:sc1}.
\end{proof}

Next, we upper bound the number of good citizens that can be fooled into accepting incorrect values.

\begin{lemma}[Upper bound on number of good citizens fooled in Global State Read]
\label{lemma:readincorrectkeys}
Let $n_e$ denote the number of good citizens that have at least one incorrect key-value pair at the end of \algoref{read-gs}. Then, for $\commnumhigh \leq 2300$ and $\eps_1 \leq 2^{-10}$, $n_e < 18$, except with negligible probability. 
\end{lemma}
\begin{proof}
Let $E_v$ denote the event where a good citizen $v$ has incorrect key-value pairs at the end of \algoref{read-gs}. From Corollary~\ref{corr:incorrect-keys}, we have probability of $E_v$ is upper bounded by $\eps_1$. 
Hence, the probability with which $E_v$ occurs for at least $18$ citizens is upper bounded by ${\commnumhigh \choose 18}\eps_1^{18}\times1$, which is $<2^{-\secparam}$ using $\commnumhigh \leq 2300$ and $\eps_1 \leq 2^{-10}$.
\end{proof}

\paragraph{Total Communication.} We calculate concrete communication of our protocol in our setting. 
For our setting, $k \leq 300,000, \deep = 30, \fanout = 25, \collide = 10, \hashsz = 10$. 
We need $\mu \tau > 7$ to achieve  $\eps_1 \leq 2^{-10}$. 
Recall the major upload and download cost in the above protocol.
\begin{tiret}
\item Upload: $ \bnum \cdot \hashsz \cdot \fanout$ (Step 8).
\item Download: $\mu \cdot k \cdot (\hashsz \cdot \deep + \collide \cdot 8)$ (Step~\ref{gsread:sc2}) $+ \tau \cdot (k/\bnum) \cdot 4 \cdot \fanout$ (Step~\ref{gsread:ex2}) $+ (\tau+\fanout)\cdot (\hashsz \cdot \deep + \collide \cdot 8)$ (Step~\ref{gsread:ex3}). %
\end{tiret}
We set $\mu = 0.015$ (i.e., $k' = 4500$), $\tau = 500$, and $\bnum = 2000$. This gives us upload of $0.5$MB and download of 1.71MB (in Step~\ref{gsread:sc2}), 7.5MB (in Step~\ref{gsread:ex2}) and 0.19MB (in Step~\ref{gsread:ex3}).

\subsection{Sampling based Merkle Tree Write}\label{sec:samplingwriteprotocol}

In our protocol, the \politicians\ will compute the updated Merkle tree $T'$. Naturally, they cannot be trusted to do this computation correctly - so we construct a protocol that will enable \citizens\ to verify the correctness of the updated Merkle tree without downloading all challenge paths. The high level idea is as follows. Let the original tree be $T$ and the updated tree be $T'$. We break $T'$ at a level, say $a$, that has $2^a$ nodes. We call this level the frontier level and the nodes at this level as frontier nodes. Now \citizens\ will obtain the values of the frontier nodes of $T'$ from a random subset of the \politicians. The \citizens\ then run a spot checking algorithm - they pick a random subset of frontier nodes and ask a \politician\ to prove the correctness of that frontier node. The \politician\ does this by showing the challenge paths (up to the frontier node) in $T'$ for all leaves that have changed and the challenge paths in $T$ for all (potentially internal) nodes that havent changed. Next, \citizens\ create exception lists with the help of a random subset of \politicians. This list denotes which frontier nodes are incorrect with the \citizens. The \citizens\ then proceed to sequentially correct the incorrect frontier nodes and then finally compute the correct root of $T'$ from the frontier nodes (The algorithm is provided in \algoref{update-gs}). We can show through a careful argument (in Lemmas \ref{lemma:updatekeyspotcheck} and \ref{lemma:incorrect-update}) that the sizes of exception lists can be bounded and \citizens\ do not accept an incorrectly updated Merkle tree $T'$, except with a small probability, which we once again factor in to the set of malicious \citizens.

We provide more details now. A citizen as well as all politicians hold a set of $\keys \subset \allkeys$ and corresponding {\em updated/new} values and the citizen wants to compute the correct new global state merkle root.
Let $k = |\keys|$.
We assume that the citizen already has the authenticated merkle root of the previous global state against which it runs the various verification checks.
Also, consider the level in the merkle tree with $2^{\cutpt}$ number of nodes. We call this level the frontier level and the nodes at this level as frontier nodes.
We describe the protocol in \algoref{update-gs}.

\newcommand{\subA}{c}

\input{algo-update-gs}

\begin{lemma}\label{lemma:updatekeyspotcheck}
For a good citizen $v$, if verification in step 4 succeeds, then at most $\tau$ values can be incorrect in Step 3, except with probability at most $\eps_2 = \left(1-\frac{\tau}{2^\cutpt}\right)^{\subA}$. Hence, $v$ learns the correct values for all frontier nodes $\fnodes$, except with probability at most $\eps_2$.
\end{lemma}
\begin{proof}
The proof of this is identical to the proof of Lemma \ref{lemma:keyspotcheck} and Corollary~\ref{corr:incorrect-keys}.
\end{proof}

Below, we will set our parameters for $\cutpt, \subA, \tau$ such that $\eps_2 \leq 2^{-10}$.

\begin{lemma}[Upper bound on number of good citizens fooled in ReadGlobalState]
\label{lemma:incorrect-update} Let $n_{\textrm{ef}}$ denote the number of good citizens that compute incorrect merkle root of $\tree'$ at the end of \algoref{update-gs}. Then, $|n_{\textrm{ef}}| < 18$, except with negligible probability for $\eps_2 \leq 2^{-10}$. 
\end{lemma}

\begin{proof}
Let $E_v$ denote the event where a good citizen $v$ computes incorrect merkle root of $\tree'$ at end of \algoref{update-gs}. This happens when $v$ 
succeeds in verification in Step 4, but greater than $\tau$ frontier nodes are incorrect in Step 3. From Lemma \ref{lemma:updatekeyspotcheck}, we have that the probability with which $E_v$ occurs is $\eps_4$. Hence, the probability with which $E_v$ occurs for $18$ citizens is bounded by ${\commnumhigh \choose 18}\eps_2^{18}\times1$, which is $2^{-\secparam}$ using $\commnumhigh \leq 2300$ and $\eps_2 \leq 2^{-10}$. 
\end{proof}

\paragraph{Some parameter settings and estimates on communication.} First, we recall the major download costs in the above protocol. Upload is tiny.
\begin{tiret}

\item Frontier nodes download: $2^{\cutpt} \cdot \hashsz \cdot \fanout$ (Steps 3 and 5a) 
\item Spot check: $\frac{k \cdot \subA}{2^a} \cdot \left(16\collide + \hashsz \cdot (2\deep - \cutpt) \right)$ (Step 4b)
\item Exception check:  $\frac{k \cdot (\tau+\fanout)}{2^a} \cdot \left(16\collide + \hashsz \cdot (2\deep - \cutpt) \right)$ (Step 5c)
\end{tiret}

For $k = 300000$, we set $\cutpt = 13, \subA = 72, \tau = 800$ to get $\eps_2 < 2^{-10}$. Also, $\deep = 30, \fanout = 25, \collide = 10, \hashsz = 10$ and communication is as follows 2.04MB (Steps 3 and 5a), 1.66MB (Step 4b) and 19MB (Step 5c).

\subsection{Savings from optimized global state protocols}
\label{app:gs-savings}
First, we compare the communication cost of our optimized global state read/write w.r.t. naive protocol that downloads the challenge paths in the Merkle tree for all relevant $\keys$, computes partial Merkle tree with new updated values to create the new root. We note that if a politician note lies in either GSRead protocol (\algoref{read-gs}) or GSUpdate protocol (\algoref{update-gs}), it is a detectable offence, i.e., there is a short proof that can be used to consistently blacklist the politician node for all honest politicians as well as honest committee members of this round. Hence, the committee can additionally sign the blacklisting of this malicious politician. With this observation, we compare our costs with the naive solution in both best case (that is all politicians behave honestly) and worst case (there exists a politician that lies either in spot check or exception list phase).

The communication cost of naive GSRead/GSUpdate is downloading $300,000$ challenge paths in a Merkle tree with $2^{30}$ leaves of size at most $10$ bytes. This amounts to $108$MB (modulo compression). In contrast, communication cost of our GSRead/Write in best case (with no exceptions) is $5.9$MB and in worst case is $32$MB (for citizens that are not fooled, i.e., all except 36 citizens). Hence, we are at least $18\times$ better in best case and $3\times$ better in worst case. Note that in the worst case, at least one malicious politician gets blacklisted.

Next, we compare the number of hash evaluations of both solutions. Verifying a challenge path requires computing $30$ hashes. Hence, for both read and update, naive solution requires $300,000 \cdot 30 \cdot 2 = 18$ million hashes to be computed. Our optimized protocol for read checks $4500$ and $5025$ challenge paths in the best case and worst case, respectively plus $2000$ hashes for computing bucket hashes. Hence, Read does $137,000$ hashes in best case and $152,750$ hashes in worst case. Our optimized protocol for global state update computes $2\deep-a$ hashes per key that has changed under a frontier node. Once it learns the correct value of all frontier nodes, it computes $2^a$ hashes to compute the root. We check $2637$ keys in best case ($123,925$ hashes) and $32,849$ keys ($1.5$ million hashes) in  worst case. Global state update addtionally needs $8192$ hashes to compute the root from the correct frontier nodes.
In total over both global state read and update, using our optimized protocols, we need $269,117$ hashes in best case and $1.7$ million hashes in worst case. Hence, we are $66\times$ and $10\times$ better in best and worst case, respectively.

%% file: algo-read-gs.tex
\begin{algorithm*}
\caption{Global State Read by a citizen node}
\label{algo:read-gs}
\begin{algorithmic}

\Require %
The citizen has the correct root of the Merkle tree corresponding to the global state (using our Get Ledger protocol in Appendix~\ref{app:get-ledger}).  Citizen and politicians hold a set of $\keys \subset \allkeys$. Let $k = |\keys|$.
\Ensure Citizen outputs values corresponding to all keys in $\keys$. 
\\

\begin{enumerate}
\item (If not already picked) Pick a subset $S = (S_1, \dotsc, S_\fanout) \subset \sn$ at random.

\item Ask for values from $S_1$. 

\item Receive a signed ordered list of $k$ values from $S_1$. (Comm. cost: Download $k \cdot 4$ bytes.) \label{gsread:getval}

\item Pick a random subset $\keys' \subset \keys$ of size $k' = \mu \cdot k$ for appropriate $\mu <1$ and send to $S_1$. (Comm. cost: Upload $\mu \cdot k \cdot 4$ bytes.) \label{gsread:sc1}

\item Receive signed list of $k'$ challenge sibling paths in the merkle tree corresponding to keys in $\keys'$ from $S_1$. (Comm. cost: Download $\mu \cdot k \cdot (\hashsz \cdot \deep + \collide \cdot 8)$ bytes.) \label{gsread:sc2}

\item Verify all the challenge sibling paths against the root of merkle tree. If verification of one of the paths fails, then it can be used as a witness to blacklist $S_1$.  \label{gsread:sc3}

\item Citizens as well as the politicians $(S_1, \dotsc, S_\fanout)$ deterministically arrange $\keys$ into $\bnum$ buckets and compute hash of keys and values in each bucket. 
Each bucket would contain $k/\bnum$ keys.
Let us denote these hashes computed by the citizen as $(\hv_1, \dotsc, \hv_\bnum)$. 

\item Upload $(\hv_1, \dotsc, \hv_\bnum)$ to $(S_1, \dotsc, S_\fanout)$. (Comm. cost: Upload $\bnum \cdot \hashsz \cdot \fanout$ bytes.)

\item A politician $S_i$ checks these hash values against its locally computed list. It creates a list of indices for which these don't match. We call this the exception list $E_i$. 

\item Citizen receives exception lists from all politicians. Let $\tau$ be a parameter fixed later such that it is guaranteed that at most $\tau$ values can be incorrect in Step~\ref{gsread:getval} with probability at least $1-\eps_1$.  An exception list is considered valid if it is of size at most $\tau$. (Comm. cost: Download tiny.) \label{gsread:ex1}

\item If exception list $E_i$ is valid, ask $S_i$ for values for all the keys in bucket whose index is in  $E_i$ and receive the same. (Comm. cost: Download $\tau \cdot (k/\bnum) \cdot 4 \cdot \fanout$ bytes.) \label{gsread:ex2}

\item Now, process these claimed incorrect buckets sequentially as follows: Pick a key for which value from Step~\ref{gsread:getval} does not match the claimed value in Step~\ref{gsread:ex2}, ask the corresponding politician for the signed challenge path to the root. Continue till the citizen corrects $\tau$ keys or exhausts all unmatched values. Moreover, if a politician provides an incorrect challenge path, consider that politician as malicious and move to the list of the next politician. (Comm. cost: Download $(\tau+\fanout)\cdot (\hashsz \cdot \deep + \collide \cdot 8)$ bytes.)  \label{gsread:ex3}

\end{enumerate}

\end{algorithmic}
\end{algorithm*}

%% file: algo-update-gs.tex
\begin{algorithm*}
\caption{Global State Update by a citizen}
\label{algo:update-gs}
\begin{algorithmic}
\Require %
The citizen has the correct root of the Merkle tree corresponding to the global state$^\dag$. Citizen and politicians hold set of $\keys \subset \allkeys$ (as well as their corresponding new values) that have to be updated. Let $k = |\keys|$.
\Ensure If a good citizen accepts then with probability at least $1-\eps$, new Merkle root is correct for that citizen (for a constant $\eps$ to be specified later.) \\

\begin{enumerate}
\item All politicians use the new key-value pairs to construct the new merkle tree. Let us denote the old tree by $\tree$ (that is authenticated by threhold number of citizens from previous committee) and the new tree by $\tree'$.

\item Citizen picks a subset $S = (S_1, \dotsc, S_\fanout) \subset \sn$ at random (if not already picked) and queries $S_1$ for frontier node values in new tree $\tree'$.

\item Citizen receives signed frontier node values in $\tree'$. Denote this set by $\fnodes$. (Comm. cost: Download $2^{\cutpt} \cdot \hashsz$ bytes.)

\item Now the citizen runs the following spot-checking algorithm for the frontier node values.

\begin{enumerate} 

\item Pick a random subset $\fnodes' \subset \fnodes$ of size $\subA$ and send to $S_1$. We will fix $\subA$ later. (Comm. cost: Upload $2^{\cutpt-3}$ bytes, tiny.)

\item For each $f \in \fnodes'$, $S_1$ computes the following: 1) for all the keys that have changed in the subtree rooted at $f$ it computes their challenge sibling path from leaf to $f$. 2) For all the internal nodes that are part of some challenge sibling path and are unchanged in $\tree$ and $\tree'$, it computes the challenge path for this (potentially internal) node in the tree $\tree$. It sends all this data signed to the citizen. Note that challenge paths for all the unchanged siblings are perfectly overlapping and can be sent together.

(Comm. cost: Download $\frac{k \cdot \subA}{2^a} \cdot ((\deep-\cutpt)\cdot \hashsz + \collide\cdot 8) + \frac{k \cdot \subA}{2^a} \cdot \left( 
\hashsz \cdot \deep + \collide\cdot 8   \right)$ bytes.)

\item Note that the citizen knows the keys that have changed as well as their location in global state merkle tree. Hence, it knows all the challenge paths that need to be provided in $\tree'$ and $\tree$.  After receiving all these challenge paths, it runs the appropriate verification procedure. If some verification fails, then it can be used as a witness to blacklist $S_1$.  

\end{enumerate}

\item Next, the citizen runs the following steps to learn exception list of incorrect frontier nodes and proofs to correct them. Let $\tau$ be a parameter to be fixed later such that the citizen would try to correct at most $\tau$ frontier nodes. 

\begin{enumerate}

\item Download the signed frontier nodes from remaining politicians $(S_2, \dotsc, S_\fanout)$.  (Comm. cost: Download $2^{\cutpt} \cdot \hashsz \cdot (\fanout-1)$ bytes.) 

\item Based on information obtained above, the citizen would create an exception list $E_i$ corresponding to each politician $S_i$. Discard a politician if the exception list is longer than $\tau$.

\item Now similar to spot-checking step above, the citizen sequentially verifies sub-trees under $\tau$ frontier nodes in the cummulative exception lists. If data provided by a politician does not verify, declare that politician as malicious. Hence, the citizen downloads data corresponding to at most $c' = \tau + \fanout$ subtrees. 
(Comm. cost: Download: $\frac{k \cdot c'}{2^a} \cdot ((\deep-\cutpt)\cdot \hashsz + \collide\cdot 8+  \hashsz \cdot \deep + \collide\cdot 8 )$ bytes.)

\end{enumerate}

\item Finally, the citizen computes the root of the new merkle tree $\tree'$ using final frontier node values. 

\end{enumerate}

\end{algorithmic}

\vspace{1em}
$^\dag$\footnotesize{This can be ensured using the signatures on the root by the previous committee. We will include the communication required to realize this in the main protocol.}
\end{algorithm*}

%% file: block-commit.tex
\section{Committing blocks}
\label{sec:block-commit}
\subsection{Protocol Description}\label{app:completeprotocol}
In this section, we describe our protocol to commit new blocks to the blockchain. This protocol builds on the sub-protocols described in the previous section and a Byzantine agreement or consensus protocol for strings \cite{AlgoRand,micaliagreement,tc84}. Define a committee round $\commround$ to start when the first honest citizen that is a member of that committee downloads the Merkle root of global state corresponding to committee round $\commround-1$. The committee round $\commround$ ends when a new block gets signed and committed by a  threshold number, say $\goodthresh$, of committee members for $\commround$. We will fix $\goodthresh$ later such that $\badnummax +36 < \goodthresh \leq \goodnum^* -36$ (because $36$ good citizens can be fooled during global state read/write, Appendix~\ref{app:gs}). Here, recall that $\badnummax \leq 772$ is the upper bound on malicious citizens (\lemmaref{maxmaliciousnodes}) and $\goodnum^* \geq 1137$ is minimum number of good citizens in any committee (\lemmaref{goodsize}). In particular, we set $\goodthresh = 850$.
We present our the protocol for committing a block number $\commround$ formally in \algoref{block-commit}. This builds on the inutitive protocol description in \S~\ref{sec:block-main} and uses our optimized protocol for global state read/write and prioritized gossip to sync on transaction pool, i.e., \tpools.
We set $\rho = 45$, the number of \tpools\ considered in a block.
As mentioned earlier, for a block $N$, committee citizens ask a fix set of politicians for transaction pools and this set of politicians picked using the hash of the round number $N$ and hash of previous block, i.e.,  $\hash(\block_{N-1})$.
Moreover, we partition the set of pending transactions across these 45 politicians to avoid overlap of transactions given in \tpools\ of different politicians.
Each transaction has an identifier consisting of identity of the transaction originator and a nonce.
We use a hash on round number $N$ and originator identity to shard these transactions randomly across politicians. Note that using round number inside hash ensures that a transaction eventually gets assigned to an honest politician helping us achieve fairness (\lemmaref{fairness}).
Finally, we set $\Delta = 350$ ensuring $\badnummax + \Delta = 1122 < \goodnum^*$ used in arguing throughput (\lemmaref{throughput}).

\input{algo-block-generate-final}

\subsection{Security Proofs}\label{sec:securityproofs}
We prove the following theorems about our protocol above.
All statements below hold except with negligible probability. Let the number of proposers in the committee be $\heavynum$.

\begin{theorem}[informal, \textbf{Safety}]
\label{theorem:safety}
After block $\commround$ gets committed, the blockchain, agreed upon by all honest citizens and politicians consists of a sequence of blocks and a authenticated global state s.t. the following holds.
\begin{enumerate}
\item This sequence of blocks $\block_1, \ldots, \block_\commround$ forms a ledger with a sequence of correct transactions.
\item Committed global state, say $\GSRoot_\commround$, contains the correct values for all the keys.
\end{enumerate}
\end{theorem}

\begin{proof}
Note that for a block to be valid, its hash needs to be signed by $\goodthresh$ number of committee members.  Lemma \ref{lemma:blockconsensus} shows that at the end of each round a block and global state consistent with existing blockchain is signed by all good citizens except $36$. First, note that $\goodnum - 36 \geq 1101 > \goodthresh = 850$ meets the thrshold needed for committed a block. Moreover, no other alternate block can satisfy the threshold because $\badnummax + 36 \leq 808 < \goodthresh = 850$.
Theorem follows from a  simple induction argument.
\end{proof}

\begin{theorem}[informal, \textbf{Liveness}]\label{theorem:liveness}
An adversary colluding with malicious politicians and citizens cannot indefinitely stall the system. Moreover, in any block-commit round, the probability that an empty block is committed is at most $\frac{1}{3}+0.017 < 0.35$.

\end{theorem}

\begin{proof}
From the safety property, it is clear that at the beginning of a round, all good citizens have a consistent view of the blockchain and global state. Next, from the design of the block-commit protocol, it is easy to see that since all good citizens talk to at least one honest politician, they will make progress and enter the consensus protocol. Moreover, \lemmaref{blockconsensus} shows that in all cases, at the end of a round, a block consistent with the existing blockchain will be signed by all good citizens (except at most $36$). This meets the threshold of block commitment, i.e., $\goodthresh = 850$. Hence, an adversary cannot stall the system indefinitely. Next, to argue liveness/progress, we need to prove that not all blocks can be empty blocks.
An empty block could be committed in one of the following cases: a) The proposer was a malicious citizen; b) The proposer was a good citizen, but all good citizens did not start the  consensus protocol with the proposal of the winner (the exception case of Lemma \ref{lemma:honest-proposer}). The former happens with probability at most $\frac{1}{3}$ because in any committee $\badnum < \commnum/3$. And, the latter probability is bounded by $0.017$ as shown in Lemma \ref{lemma:honest-proposer}. Hence the lemma follows.

\end{proof}

We also prove bounds on throughput Lemma~\ref{lemma:throughput} and fairness Lemma~\ref{lemma:fairness}.  First, we prove lemmas needed to prove safety and liveness.

\paragraph{Lemmas needed to prove the above theorems.}
Here, we will work with concrete parameters for our committee. From Lemma \ref{lemma:commsize}, we know that $1700 \leq \commnum \leq 2300$.
From Lemma \ref{lemma:gap}, we have $\gap \geq 1$, that is, for any committee $\commnum > 3\badnum +1$ and condition in Property~\ref{prop:consensus} is satisfied.
Also, from \lemmaref{goodsize}, $\goodnum \geq 1137$ and from \lemmaref{maxmaliciousnodes}, $\badnum \leq \badnummax \leq 772$.

We will use the following property of the consensus protocol.

\begin{property}[\cite{AlgoRand, micaliagreement, tc84}]
\label{prop:consensus}
If the consensus protocol between $\commnum$ parties s.t. $\commnum \geq 3t+1$ results in all honest parties outputting a non-null value, then at least $\lceil(\commnum-1)/3\rceil$ honest parties entered the consensus protocol with that initial value. Here, $t$ denotes number of malicious parties.
\end{property}

In the following two lemma, we will use the above property of the consensus protocol to analyse our block commit protocol separately for the case when the winning proposer is honest or malicious.

\begin{lemma}[Honest Good Proposer]
\label{lemma:honest-proposer}
If a good citizen is the winning proposer, then, except with probability $0.017$, all good citizens will enter the consensus protocol with the proposal of the winner and output the proposal of the winner as output of the consensus protocol.
\end{lemma}
\begin{proof}
In the case of a proposer who is a good citizen, since any \tpool\  is picked in the proposal iff it is present in the witness lists of $\badnummax+\Delta$ citizens, it is definitely the case that at least $\Delta$ good citizens have this \tpool\ after Step~\ref{first-down}. Now, in Step~\ref{first-up}, the probability that \tpool\ is not re-uploaded by any good citizen to any honest politician is bounded by
$\left(\frac{\rho-5}{\rho}+\frac{5\corrsn}{\rho}\right)^\Delta$. Applying a union bound over all $\rho$ \tpools, the probability that there exists such a \tpool\ is bounded by $\rho\left(\frac{\rho-5}{\rho}+\frac{5\corrsn}{\rho}\right)^\Delta$. When we set $\Delta = 350$, and $\rho = 45$, this probability is $\leq  0.017$.
Now, from the correctness of our priorotized gossip protocol, it follows that if a \tpool\ is uploaded by a good citizen, it is learnt by all honest politicians and hence, by all good citizens during the download in Step~\ref{second-down} of \algoref{block-commit}. Hence, all good citizens will enter consensus with this proposed block. The lemma then follows from the correctness of the consensus algorithm.
\end{proof}

\begin{lemma}[Malicious Proposer]
\label{lemma:malicious-proposer}
If a malicious citizen is the winning proposer, and consensus results in non-null, then all good citizens will be able to download the \tpools\ committed in the proposal, except with negligible probability. %
\end{lemma}
\begin{proof}
By Property \ref{prop:consensus}, we know that if the consensus results in non-null, then at least $\lceil(\commnum-1)/3\rceil$ number of good citizens started the consensus with that initial value and had all the required \tpools. With this, for step~\ref{second-up}, the probability that a \tpool\ is not re-uploaded by any good citizen to an honest politician is bounded by $\left(\frac{\rho-10}{\rho}+\frac{10\corrsn}{\rho}\right)^{\lceil(\commnum-1)/3\rceil}$ and the probability that there exists such a \tpool\ among the all of $\rho$ \tpools\ is bounded by $\rho\left(\frac{\rho-10}{\rho}+\frac{10\corrsn}{\rho}\right)^{\lceil(\commnum-1)/3\rceil}$. Since $\commnum\geq 1700$, we have $\lceil(\commnum-1)/3\rceil \geq 566$, and hence, for $\rho = 45$ we have that this probability is negligible. Once the politicians gossip the \tpools, all honest politicians will have the \tpools\ and hence, every good citizen can download all \tpools\ in Step~\ref{third-down} of \algoref{block-commit}.

\end{proof}

\begin{lemma}[Block Consensus]\label{lemma:blockconsensus}
Except with negligible probability, (i)
at the end of the block commit protocol, all good citizens except $36$
sign the same block hash and new global state root and (ii) moreover, this block is consistent with the entire block chain and global state.
\end{lemma}

\begin{proof}
We consider the following two cases:
\begin{itemize}

\item (Good citizen is winner of block proposal): From Lemma \ref{lemma:honest-proposer}, all good citizens hold all the required \tpools\ and agree on honest proposer's block except with probability $0.017$. When this happens, all good citizens and honest politicians construct the same block.
Moreover, by \lemmaref{readincorrectkeys} and \lemmaref{incorrect-update}, we know that at least $\goodnum - 36 \geq 1101$ good citizens are able to do global state read/write operations correctly and compute correct new global state root. Finally, good citizens follow a deterministic procedure to create a valid block.
Hence, all good citizens except $36$  will sign the same block hash and new global state root.

With probability $0.017$,  all good citizens do not enter the consensus with proposer's block. This case is identical to malicious proposer discussed below.

\item (Malicious citizen is winner of block proposal): In this case, if the consensus results in an empty block, then all good citizens ($\goodnum \geq 1137$) will sign the empty block and previous global state root.
If the consensus results in non-null, then from Lemma \ref{lemma:malicious-proposer}, all good citizens will be able to download the \tpools\ committed in the proposal (except with negligible probability) in Step~\ref{third-down}. Hence, using same arguments as above, all good citizens except $36$ will will compute the same block hash and new global state root and sign this. This number is clearly more than $\goodthresh$.

\end{itemize}

To prove consistency of the block and global state with the entire block chain, observe that at most $36$  good citizens could have obtained either an incorrect $(\key,\val)$ pair or an incorrect frontier node $\fnodes$ (Lemmas \ref{lemma:readincorrectkeys} and \ref{lemma:incorrect-update}). All other  good citizens %
download the correct global state, verify all transactions honestly and only sign blocks with valid transactions that are consistent with the global state. They also update the global state honestly.

\end{proof}

\begin{lemma}[Lower bound on throughput]
\label{lemma:throughput}
The system does not stall in the presence of adversarial behavior and in expectation, committed blocks will have at least $0.65(1-\corrsn)\rho$ \tpools\ in them.
\end{lemma}

\begin{proof}

First, note from Theorem~\ref{theorem:liveness} (Liveness), we have that the system does not stall and a non-empty block gets committed to the blockchain with probability at least 0.65, and in particular, this happens with an honest citizen as the proposer. Now, note that malicious politicians can refuse to provide \tpools, provide them to only a few honest mobile nodes, or provide \tpools\ that have transactions that do not verify (the last attack however, will get a politician blacklisted because \tpools\ are signed.). However, since, $T_1, \ldots, T_\rho$, which are the politicians that would provide the transactions for the block in round $\mathsf{Rnd}$ are chosen at random, the expected number of honest politicians in this set of $\rho$ \tpools\ is $(1-\corrsn)\rho$. Furthermore, the \tpools\ corresponding to these politicians are non-intersecting from the way transactions are assigned to politicians. For all honest politicians' \tpools\ in this set, the condition in Step~\ref{check-wit} of \algoref{block-commit} will be met since $\goodnum \geq \goodnum^* > \badnummax + \Delta$. Hence, the protocol would result in committing this block with at least $(1-\corrsn)\rho$ \tpools\ in them. Combining this with the probability of a non-empty block, we obtain the proof of the lemma.

\end{proof}

\begin{lemma}[informal, \textbf{Fairness}]
\label{lemma:fairness}
All valid transactions will eventually be committed.
\end{lemma}
\begin{proof}

As in the proof of Lemma \ref{lemma:throughput}, from Theorem~\ref{theorem:liveness}, we know that the system does not stall and that non-empty blocks get committed with probability at least 0.65. A malicious block proposer could only pick \tpools\ provided by malicious politicians or always creates an empty block. However, for a block constructed by an honest proposer, since the politicians that would provide the transactions for the block in round $\mathsf{Rnd}$ are chosen at random, eventually every valid transaction held by honest politicians will be selected by an honest proposer. Hence, we can argue that all valid transactions are eventually committed.
\end{proof}

%% file: algo-block-generate-final.tex
\begin{algorithm*}
\caption{Block Commit}
\label{algo:block-commit}
\begin{algorithmic}[1]
\Require Committee members for block $\commround$ hold the authenticated $\GSRoot_{\commround-1}$ and $\hash(\block_{\commround-1})$ signed by the corresponding committee.
 Let $V = v_1, \ldots, v_\commnum$ be the committee for block $\commround$ with $V \supset U = u_1, \ldots, u_{\heavynum}$ as the proposers.  Let $T_1, \ldots, T_\rho$ be the politicians that would provide the transactions in $\commround$. 

\Ensure $\goodthresh$ number of citizens of block $\commround$ will  sign and ``commit'' a block (also containing $\hash(\block_{\commround-1})$) along with the updated root of the global state at the end of the protocol.

\begin{enumerate}
\item Download transaction pools and upload witness lists: Every committee members $v \in V$ does the following:
\begin{enumerate}
\item 
For all $i \in [\rho]$ download a signed \tpool\ (of appropriate size) $(\mem_i, \sigma_i)$ from $T_i$, where $\sigma_i$ is the signature on $(\hash(\mem_i), \mathsf{Rnd})$ by $T_i$. \label{first-down}
\item Pick a random subset $S\p{v}  \subset \sn$ of size $\fanout$ for the rest of the protocol. (The committee member talks to the same politicians for the whole protocol).
\item Create a witness list $\wl_v$ of length $\rho$ as follows: $\wl_v[i] = (\hash(\mem_i), \sigma_i)$ if it has $i^{th}$ \tpool, and $\wl_v[i] = \bot$ otherwise. We call $(\hash(\mem_i), \sigma_i)$ a {\em commitment}.
\item It uploads $(\wl_v, \sigma_v, \vrf_v)$ to all $S\p{v}$, where $\sigma_v$ is signature on $\hash(\wl_v)$ by citizen $v$.
\end{enumerate} 

\item Politicians gossip on these signed witness lists and VRFs using full broadcast. Within these witness lists, if a politician finds two different valid commitments by another politician, it constitutes a proof of malicious behavior by that politician that can be used to blacklist the malicious politician. %

\item Block Proposal by proposers and upload of transaction pools by citizens: 
A proposer $u \in U$ does the following:
\begin{enumerate}
\item \label{blah} Downloads signed witness lists and VRFs of all committee members from $S\p{u}$. 
\item Picks commitments that are non-null in  witness lists of at least $\badnummax+\Delta$ committee members.  \label{check-wit}
Denote these by $((\hash(\mem_{j_1}), \sigma_{j_1}), \ldots, (\hash(\mem_{j_{\rho}}), \sigma_{j_{\rho}}))$.
\item Upload the following to $S\p{u}$: i) $\vrf_u$ (a proof that it is a proposer for this round) 
ii) $\idlist_u = ((\hash(\mem_{j_1}), \sigma_{j_1}), \ldots, (\hash(\mem_{j_{\rho}}), \sigma_{j_{\rho}}))$ iii) Signature on $\hash( \idlist_u)$. We denote these by $\proposal_u$.

\end{enumerate}
In parallel, every  committee member $v$ picks $5$  indices $i_1,i_2, i_3, i_4, i_5 \in [\rho]$ at random s.t. for all $j \in [5]$, $\wl_v[i_j] \ne \bot$ and uploads all $(\mem_{i_j}, \sigma_{i_j})$ to a random politician. \label{first-up}%

\item Politicians sync on the block proposals, i.e., $\{\proposal_u\}_{u \in U}$ using full broadcast and sync on \tpools\ uploaded by last step using our prioritized gossip protocol (see \S~\ref{sec:forced-truth-gossip-main}).

\item Compute winner of block proposal and initial value in consensus. Every committee members $v \in V$ does the following:
\begin{enumerate}
\item Download block proposals, i.e., $\{\proposal_u\}_{u \in U}$ from $S\p{v}$. Winner $\win(v)$ is the one with least value of $\vrf_u$.
\item Try downloading \tpools\ that are missing after  Step~\ref{first-down} from $S\p{v}$.  If successful in downloading all \tpools\ committed in $\idlist_{\win(v)}$, set the initial value for consensus protocol $\initial_v = \proposal_{\win(v)}$, else set $\initial_v = \nullv$. \label{second-down}

\end{enumerate}

\item Consensus and second round of upload-download of \tpools\

\begin{enumerate}
\item Committee members run the consensus protocol~\cite{AlgoRand, micaliagreement, tc84}, where committee member $v$ enters consensus with value $\initial_v$,  to reach consensus on value $\out$.
\item In parallel, every committee member $v$ picks $10$ out of $\rho$ \tpools\ it has at random and uploads them to a random politician. \label{second-up}
\item Politicians sync on the \tpools\ uploads using our protocol for prioritized gossip.
\item Committee members try downloading the \tpools\ that are still missing after Step~\ref{second-down}.  \label{third-down}
\end{enumerate}

\end{enumerate}
\hspace{5in}~~~~~~Contd....
\algstore{myalg}
\end{algorithmic}
\end{algorithm*}

\begin{algorithm*}                     
\begin{algorithmic} %
\algrestore{myalg}
\Require Continuation of \algoref{block-commit}
\Ensure Citizens commit a block $\block$.
\begin{enumerate}  \setcounter{enumi}{6}

\item Block commit: There are following two cases:

Case 1: $\out = \nullv$: Every committee member $v \in V$ commits to an empty block $E$ by uploading $(\hash(E), r, \sigma_v)$ to $S\p{v}$, where $r = \GSRoot_{\commround-1}$ as the root of current global state merkle tree, $\sigma_v$ is the signature on $(\hash(E), r, \mathsf{Rnd})$. 

Case 2: $\out = \proposal_{u^*}= (\vrf_{u^*}, H_{u^*}, \idlist_{u^*}, \sigma_{u^*})$,  $u^* \in U$. Politicians create the block $\block_{u^*}$ using \tpools\ in $\idlist_{u^*}$ and delta merkle tree for the updated global state to be used for \algoref{update-gs}. Every committee members $v \in V$ does the following: 
\begin{enumerate}
\item \label{validate-gs} Create a valid block$^\dag$ $\block_{u^*}$ by following the same deterministic procedure over \tpools: It verifies the signatures on transactions, reads the values of keys being updated from the last committed global state using \algoref{read-gs}, orders transactions, and applies them on global state to create a block $\block_{u^*}$ of valid transactions. Block also includes $\hash(\block_{\commround-1})$ for cryptographic chaining property. 
\item \label{new-gs} Run \algoref{update-gs} to compute the root $r_v$ of the new global state merkle tree.
\item Upload $(\hash(\block_v), r_v, \sigma_v)$ to $S\p{v}$, where $\sigma_v$ is the signature on $(\hash(\block_v), r_v, \commround)$.%
\end{enumerate}

\item Final block commit: Politicians sync on the data uploaded in last step and accumulate signatures on same values. Final committed block and the global state is the one with at least $\goodthresh$ signatures of committee members for block $\commround$.
\end{enumerate}
\end{algorithmic}

\vspace{1em}
$^\dag$\footnotesize{We prove that if the consensus results in non-null, then all good citizens have the required mempools by this step.}
\end{algorithm*}